\documentclass[runningheads,envcountsame]{llncs}
\usepackage{url}
\usepackage{stmaryrd}
\usepackage{graphicx}
\usepackage{amsmath,amsfonts,amssymb,cmll,bbm}
\usepackage{shortlst}
\usepackage[all]{xy}
\usepackage{subfigure}

\usepackage[mathletters]{ucs}
\usepackage[utf8x]{inputenc}
\usepackage{times}
\usepackage[T1]{fontenc}
\usepackage[unicode,hyperindex,colorlinks=true,citecolor=red,linkcolor=blue]{hyperref}

\usepackage[english]{babel}

\usepackage{sets}

\newcommand\cC{\mathcal{C}}

\newcommand\cP{\mathcal{P}}

\newcommand\cS{\mathcal{S}}

\newcommand{\st}{;\ }

\usepackage{cmll}

\newcommand{\tensor}{\otimes}

\DeclareMathAlphabet{\mathLLrm}{T1}{cmr}{m}{n}

\usepackage{bussproofs}
\input{typing}

\newcommand{\lllbracket}{\left\llbracket\kern-0.3em\left\llbracket}
\newcommand{\rrrbracket}{\right\rrbracket\kern-0.3em\right\rrbracket}

\makeatletter
\newcommand{\sem}{\@ifnextchar[{\@semtwo}{\@semone}}
\def\@semone#1{\left\llbracket#1\right\rrbracket}
\def\@semtwo[#1]#2{\relSem{#2}{#1}}
\makeatother
\newcommand{\relSem}[2]{\left\llbracket#1\right\rrbracket_{#2}}

\makeatletter
\newcommand{\semm}{\@ifnextchar[{\@semmtwo}{\@semmone}}
\def\@semmone#1{\lllbracket#1\rrrbracket}
\def\@semmtwo[#1]#2{\relSemm{#2}{#1}}
\makeatother
\newcommand{\relSemm}[2]{\lllbracket#1\rrrbracket_{#2}}

\makeatletter
\renewcommand\subsubsection{\@startsection{subsubsection}{3}{\z@}%
                       {-5\p@ \@plus -4\p@ \@minus -4\p@}%
                       {-0.5em \@plus -0.22em \@minus -0.1em}%
                       {\normalfont\normalsize\bfseries\boldmath}}

\makeatother
\renewcommand{\qed}{
\hfill$\square$}

\usepackage{float}	
\floatstyle{boxed}
\restylefloat{figure}
\newcommand{\vip}[1]{\textbf{#1}}
\newcommand{\annex}[1]{\hfill \textit{\small (#1})}
\newcommand{\refannex}[1]{ 
  \newcounter{#1}
  \setcounter{#1}{\value{theorem}}
  \addtocounter{#1}{-1}}

\newcommand{\V}{\mathcal{V}}
\newcommand{\W}{\mathcal{W}}
\renewcommand{\P}{\mathcal{P}}
\renewcommand{\C}{\mathcal{C}}

\renewcommand{\tilde}{\widetilde}
\newcommand{\supp}[1]{\lvert#1\rvert}
\renewcommand{\st}{\! \mid\! }
\newcommand{\M}[1]{\mathcal{M}_{\mathrm{fin}}(#1)}

\renewcommand{\nat}{\mathcal{N}}
\renewcommand{\int}{\mathbb{N}}

\newcommand{\bool}{\mathcal{B}}
\newcommand{\bweb}{\mathbb{B}}
\newcommand{\ort}[1]{#1^\bot}
\newcommand{\true}{\mathtt{T}}
\newcommand{\false}{\mathtt{F}}

\newcommand{\lparr}{\,\otimes_\varepsilon\,}

\newcommand{\ltens}{\,\tilde\otimes\,}

\newcommand{\Field}{\mathbbm{k}}
\newcommand{\Lc}[2]{\mathcal{L}_{\mathrm{c}}(#1,#2)}
\newcommand{\EvF}[1]{\Field\langle#1\rangle}
\newcommand{\Tot}[1]{\mathcal{T}\<#1\>}

\newcommand{\Part}[1]{\mathcal{P}(#1)}
\newcommand{\PartF}[1]{\mathcal{P}_{\mathrm{fin}}(#1)}
\newcommand{\dort}[1]{#1^{\bot\bot}}
\newcommand{\F}{\mathcal{F}}

\newcommand{\RF}{\mathbf{RelFin}}
\newcommand{\LF}{\mathbf{LinFin}}
\newcommand{\TF}{\mathbf{TotFin}}
\newcommand{\GLF}{\mathcal G(\mathbf{LF})}

\newcommand{\Poln}[2]{\mathrm{Pol}^{#1}(#2)}
\newcommand{\Polk}[1]{\mathrm{Pol}(#1)}
\newcommand{\Polc}[1]{\tilde{\mathrm{Pol}}(#1)}
\newcommand{\Pol}[1]{\Field\left[#1\right]}
\newcommand{\Ser}[1]{\Field\left(#1\right)}

\newcommand{\Lat}{\mathbf\Lambda_{\mathrm{at}}}
\newcommand{\Laff}{\mathbf\Lambda_{\mathrm{bar}}}

\newcommand{\bb}[1]{\mathbf #1}
\newcommand{\ITE}[3]{\mathtt{if}\ #1\ \mathtt{then}\ #2\ \mathtt{else}\ #3 }

\newcommand{\gus}{\mathtt{Gus}}
\newcommand{\por}{\mathtt{POr}}

\newcommand{\pol}[1]{#1^\bullet}
\newcommand{\dpol}[1]{#1^{\bullet\bullet}}
\newcommand{\T}{\mathcal{T}}
\newcommand{\pair}[2]{[\,#1\,,\,#2\,]}
\newcommand{\triple}[3]{[\,#1\,,\,#2\,,#3\,]}

\newcommand{\aff}{\mathrm{aff}\,}
\newcommand{\caff}{\overline{\mathrm{aff}}\,}
\newcommand{\ann}{\mathrm{ann}\,}
\newcommand{\ldir}{\mathrm{dir}\,}
\newcommand{\<}{\langle}
\renewcommand{\>}{\rangle}

\newcommand{\adj}[1]{#1^\ast}
\newcommand{\trans}[1]{{\vphantom{#1}}^{\mathit t}{#1}}
\renewcommand{\exp}[1]{#1^\oc}

\begin{document}
\title{Algebraic totality, towards completeness.}  
\author{Christine Tasson }
\institute{Preuves, Programmes, Syst\`emes}
\maketitle
\begin{abstract}
  \emph{Finiteness  spaces} constitute a  categorical model  of Linear
  Logic whose objects  can be seen as linearly  topologised spaces, (a
  class of topological vector  spaces introduced by Lefschetz in~1942)
  and  morphisms   as  continuous  linear  maps.    First,  we  recall
  definitions of finiteness spaces and describe their basic properties
  deduced  from the  general  theory of  linearly topologised  spaces.
  Then we  give the interpretation  of LL with an  algebraic approach.
  Second,  thanks  to  separation  properties,  we  can  introduce  an
  algebraic notion of totality  candidate in the framework of linearly
  topologised spaces: a totality candidate is a closed affine subspace
  which does  not contain  $0$.  We show  that finiteness  spaces with
  totality candidates constitute a model of classical LL.  Finally, we
  give    a   barycentric    simply   typed    lambda-calculus,   with
  booleans~$\bool$   and  a   conditional  operator,   which   can  be
  interpreted  in   this  model.    We  prove  completeness   at  type
  $\bool^n\to\bool$ for every $n$ by an algebraic method.
\end{abstract}

\section*{Introduction}
In the 80's, Girard has been  led to introduce linear logic (LL) after
a denotational investigation  of system F. The basic idea  of LL is to
decompose  the intuitionistic  implication into  a linear  one  and an
exponential modality.  Many intuitions  behind LL are rooted in linear
algebra  and relate  algebraic  concepts with  operational ones.   For
instance, a linear function in the  LL sense is a program (or a proof)
which uses  its argument exactly once  and LL shows that  this idea is
similar to linearity in the  algebraic sense. Can we use vector spaces
and linear maps for  interpreting LL? In the exponential-free fragment
of LL,  this is quite  easy, since all  vector spaces can  stay finite
dimensional:  it  is  sufficient   to  take  the  standard  relational
interpretation of a  formula (which is a set) and  to build the vector
space  generated  by  this  set.  However,  the  exponential  modality
introduces  infinite  dimension  and   some  topology  is  needed  for
controlling the size of dual spaces.  Indeed, we want all spaces to be
reflexive, that  is, isomorphic to their second  dual, because duality
corresponds to linear negation which is involutive.

There  are  various  ways  for defining  interpretations  with  linear
spaces. Among them, the  interpretations based on linearly topologised
spaces~\cite{blute:lls,ehrhard:fs} have  the feature of  not requiring
any topology on  the field $\Field$. This is  quite natural, since the
topology  of  the  field   is  never  used  for  interpreting  proofs.
Introduced first by Lefschetz in~\cite{lefschetz:at}, these spaces are
geometrically quite unintuitive (their basic opens are linear subsets
whereas usual  basic opens are  balls). They provide  nevertheless the
simplest setting where  formul\ae\ of LL can be  seen as (topological)
linear  spaces  as shown  by  Ehrhard  when  he introduced  finiteness
spaces~\cite{ehrhard:fs}. 

There are two ways of considering finiteness
spaces:\\
\emph{Relational finiteness spaces}: they  can be seen as a refinement
of the relational semantics of linear logic, in which the semantics of
proofs is the same as the standard one (proofs are
interpreted as relations).\\
\emph{Linear  finiteness  spaces}:   given  a  field,  any  relational
finiteness space  gives rise to  a linearly topologised  vector space.
The  category  of  linear  finiteness  spaces  and  continuous  linear
functions  constitutes a  model  of linear  logic.   Besides, a  linear
finiteness space and its dual with  the evaluation as pairing is a Chu
space~\cite{chu:acc}.  The  proofs of linear logic  are interpreted as
multilinear  hypocontinuous maps  (hypocontinuity is  between separate
continuity    and   continuity).     Our    description   of    proofs
(c.f. Annex~\ref{annex:interp}) is closed to that of Game Categories
of Lafont and Streicher~\cite{lafont:gsl}.

For describing these categories, we use the duality presentation whose
importance  has been emphasised  by models  of Classical  Linear Logic
such   as  phase   semantics.   Even   the  definition   of  coherence
spaces~\cite{girard:ll},  usually  described  by  means  of  a  binary
symmetric and reflexive \emph{coherence relation} $\coh_X$ on  a set
$\supp  X$, can  be reformulated  through  duality~\cite{barr:aca}. We
will freely use the  terminology of~\cite{hyland:goll} --- a survey of
the  different duality presentations  and in  particular of  models of
linear  logic by double  orthogonal. The  \emph{partial orthogonality}
between the subsets $c$ and $c'$ of $\supp X$ is given by
\begin{equation*}
  c\perp_{\mathrm{coh}} c'\iff \sharp(c\cap  c')\leq 1.
\end{equation*}
A coherence space can then be seen as a pair $X=(\supp X,\C(X))$ where
$\supp X$ is a set and $\C(X)$ is a subset of $\P(\supp X)$ (the powerset
of $\supp  X$), which is $\perp_{\mathrm{coh}}$-closed,  that is equal  to its second
dual       for       the       duality      induced       by       the
orthogonality:~$\C(X)=\C(X)^{\perp\perp}$.   The  elements of  $\C(X)$
are  the  \emph{cliques}  of  $X$,  those  of  $\ort{\C(X)}$  are  the
anticliques.   The category  of  coherence spaces  and  cliques is  an
\emph{orthogonality category}.

The category of relational finiteness spaces and finitary relations (a
refinement  of standard  relations) also  constitute  an orthogonality
category with  respect to  the \emph{finite orthogonality}  defined as
follows: let $u$ and $u'$ be two sets,
\begin{equation*}
  u\perp_{\mathrm{fin}} u' \iff \sharp (u\cap u') < \infty.
\end{equation*}
A  relational finiteness  space is  a pair  $X=(\supp  X,\F(X))$ where
$\supp X$  is a countable set  and the set  $\F(X)$ of \emph{finitary}
parts is $\perp_{\mathrm{fin}}$-closed.  The elements of $\ort{\F(X)}$
are the \emph{antifinitary} parts. We carry the relational finiteness in
the linear world by  considering the linear subspace of $\Field^{\supp
  X}$  generated by  finitary  linear combinations,  that is  families
$x=(x_a)_{a\in\supp X}$ such that their support $\supp x$ is finitary.
The     finite      orthogonality     between     supports     ($\supp
x\perp_{\mathrm{fin}}\supp  x'$)  implies that  the  pairing between  a
finitary  linear   combination  $x$  and  antifinitary   one  $x'$  is
well-defined:
\begin{equation*}
   \<x',x\>=\textstyle\sum_{a\in        \supp        X}x'_a\,        x_a
   =\textstyle\sum_{a\in\supp x\cap\supp x'} x'_a\, x_a
  \text{ is a finite sum.}
\end{equation*}

The  notion  of  totality,  introduced by  Girard~\cite{girard:sf}  in
denotational   semantics,  is  used   for  interpreting   proofs  more
closely. It often gives the  means to prove completeness results as in
Loader~\cite{loader:llt}. Girard-Loader's totality  is described by an
orthogonality  category up  to a  slight modification  of  the partial
orthogonality:
\begin{equation*}
  u\perp_{\mathrm{tot}} u' \iff \sharp (u\cap u')=1.
\end{equation*}
\noindent A totality candidate is then a subset $\Theta(X)$ of $\P(X)$
such  that $\Theta(X)$  is $\perp_{\mathrm{tot}}$-closed.   A totality
space\footnote{The additional conditions that are actually required
  in~\cite{loader:llt} are not essential for our purpose.}
is a pair $(\supp X,\Theta(X))$ where $\theta(X)$ is a totality candidate.

The  notion of  totality  can be  adapted  to linear  spaces.  So,  we
introduce the polar orthogonality:
\begin{equation*}
  x\perp^\bullet x' \iff \<x',x\>=1.
\end{equation*}
Because we  are working in  a linear algebra  setting, we are  able to
give  a  simple  characterisation  of  totality  candidates,  that  is
polar-closed  subspaces  of   linear  finiteness  spaces:  a  totality
candidate of  a finiteness  space is either  the space itself,  or the
empty set, or  any topologically closed affine subspace  that does not
contain   $0$  and  which   is  topologically   closed.   We   get  an
orthogonality category  whose objects are pairs of  a finiteness space
and  a  totality  candidate;  and  whose maps  are  continuous  linear
functions that preserve  the totality candidates.  This is  a model of
LL.

Since totality  candidate are affine spaces,  it is natural  to add an
affine construction to LL: we thus introduce \emph{barycentric LL}. We
address then  the question  of completeness: is  it the case  that any
vector in the totality candidate of a formula is the interpretation of
a proof of  this formula?  Restricting our attention  to a barycentric
version of  the simply  typed lambda-calculus (extended  with booleans
$\bool$  and a conditional  operator), we  prove completeness  at type
$\bool^n\Rightarrow\bool$ for all $n$, by an algebraic method.

\vip{Outline.}  We start Section~\ref{section:fs} with generalities on
finiteness spaces  at both levels.   Then, we give  several properties
inherited  from   linearly  topologised  spaces,   in  particular,  we
introduce separation  results that are fundamental in  the sequel.  We
describe  the interpretation  of linear  logic proofs  into finiteness
spaces                relying               on               Ehrhard's
results~\cite{ehrhard:fs,ehrhard:fs-presheaves}.                     In
Section~\ref{section:total}, after  having defined totality candidates
and the associated total  orthogonality category, we study barycentric
$\lambda$-calculus.  Finally, in Section~\ref{section:comp}, we tackle
the completeness  problem and give  a positive answer for  first order
boolean types.

\section{Finiteness spaces}
\label{section:fs}


\subsection{Relational finiteness spaces}
Let   $\mathbb  A$  be  a  countable   set.   The  finite
orthogonality is defined by:
\begin{equation*}
  \forall u,\,u'\in\mathbb A,\ u\perp u'\iff u\cap u'\text{ finite}.
\end{equation*}

As usual,  the orthogonal  of any $\F\subseteq\Part{\mathbb{A}}$  is $
\ort{\F}=\left\{u'\subseteq  \mathbb  A\st  \forall u\in  \F,\  u\perp
  u'\right\}$ and $\F$ is orthogonally closed whenever  $\dort{\F}=\F$.

\begin{definition}\label{def:finrel}
  A \vip{relational finiteness space} is a pair $A=(\supp A , \F(A))$
  where  the  \vip{  web}  $\supp  A$  is a  countable  set  and  the
  \vip{finitary    subsets}   $\F(A)\subseteq\Part{\supp    A}$   are
  orthogonally closed.  We call $u\in\F(\ort A)$ \vip{antifinitary}.
  Let  $A$  and  $B$  be  relational finiteness  spaces.   A  finitary
  relation $R$  between $A$ and  $B$ is a  subset of $A\times  B$ such
  that
  \vskip-1.5em
  \begin{eqnarray*}
    \forall u\in\F(A),\ R\cdot u&=&\left\{b\in\supp B\st \exists a\in
      u,\,(a,b)\in R\right\}\in\F(B),\\
    \forall v'\in\ort{\F(B)},\ \trans{R}\cdot v'&=&\left\{a\in\supp A\st 
      \exists b\in v',\,(a,b)\in R\right\}\in\ort{\F(A)}.
  \end{eqnarray*}
  Let  us call  $\RF$ the  category whose  objects are  the relational
  finiteness spaces and whose maps are the finitary relations.
\end{definition}

Every  finite subset  of  a  countable set  $\mathbb  A$ is  finitary.
Therefore, there  is only  one relational finiteness  space associated
with a finite web (any subset is finitary).

Let    $\mathcal   F,\,\mathcal    G\subseteq\Part{\mathbb{A}}$.    If
$\mathcal{F} \subseteq  \mathcal{G}$ then $\ort{\mathcal{G}} \subseteq
\ort{\mathcal{F}}$.    Besides,  $\mathcal{F}\subseteq  \dort{\mathcal
  F}$,    hence    $\mathcal{F}^{\bot\bot\bot}=\ort\F$.     Therefore,
$(\mathbb A, \dort \F)$ is always a finiteness space.

Let     $A$    be    a     relational    finiteness     space,    then
$\ort{(\ort{\F(A)})}=\F(A)$. Thus,  the \vip{orthogonal} $\ort  A$ of
$A$ defined  to be $(\supp A,\ort{\F(A)})$ is  a relational finiteness
space whose  orthogonal $\dort{A}=(\supp A,\dort{\F(A)})$  is equal to
$A$.
\begin{example}\label{ex:bool}
  \emph{Booleans.}   The   relational  finiteness  space   $\bool$  is
  associated     with    the     web    with     two     elements    $
  \bweb=\{\true,\false\}$.           Every          subset          is
  finitary:~$\F(\bool)=\Part{\bweb}$. \\
  \emph{Integers.}  The  web $\int$  of integers, associated  with the
  finite  subsets $\PartF{\int}$  constitutes  a relational  finiteness
  space    denoted   $\nat$.     Its   orthogonal    $\ort{\nat}$   is
  $(\int,\Part{\int})$.
\end{example}

\subsection{Linear finiteness spaces}

{\bf Notations:} In  the sequel $A$, $(A_i)_{i\leq n}$  and $B$ range
over relational finiteness spaces. The field $\Field$ is discrete and infinite 
(i.e. every subset of $\Field$ is open). We  handle standard notions of
linear algebra using the notations:
\begin{runitemize}
\item[ $E^\ast$] is the space of linear forms over
  $E$, 
\item[$E'$] is the topological dual of $E$,
\item[$\<x^\ast,x\>$] is $x^\ast(x)$ if $x\in E$ and
  $x^\ast\in E^\ast$,
\item[$\ker_E(x^\ast)$]   is  the  kernel  of
  $x^\ast\in E^\ast$,
\item[$\ann_{E^\ast}(x)$]     (resp.\
  $\ann_{E'}(x)$) is  the subspace of  $E^\ast$ (resp.\ $E'$)  of linear
  forms (resp.\ continuous linear forms) which annihilate $x$,
\item[$\aff(D)$]  (resp.\  $\caff(D)$) is  the
  affine hull (resp.\ affine closed) of a subset $D$.
\end{runitemize}

Any relational finiteness space $A$  gives rise to a linear finiteness
spaces   $\EvF{A}$  which   is  a   subspace  of   the   linear  space
$\Field^{\supp A}$:
\begin{definition}
  For   every
  $x\in\Field^{\supp A}$,  let $\supp x=\{a\in\supp  A\st x_a\neq 0\}$
  be the \vip{support} of $x$.
  The \vip{linear finiteness space} associated with $A$ is
  $
    \EvF{A}=\{x\in \Field^{\supp A}\st \supp x\in \F(A)\}.
  $
\end{definition}

With each $a\in\supp A$,  we associate a basic vector $e_a\in\EvF{A}$.
Notice  that  $\EvF A$  is  generated  by  the \emph{finitary}  linear
combinations of basic vectors (and not by finite ones).

Each linear finiteness  space can be endowed with  a topology which is
induced  by  the  antifinitary  parts  of  the  underlying  relational
finiteness space:
\begin{definition}
 For   every
  ${J'}\in\ort{\F(A)}$, let  us call $  V_{J'}=\{x\in\EvF{A} \st \supp
  x\cap {J'}=\emptyset\}$ a \vip{fundamental linear neighbourhood} of
  $0$.  A  subset $U$ of  $\EvF A$ is  open if and only  if for
  each   $x\in   U$    there   is   $J'_x\in\ort{\F(A)}$   such   that
  $x+V_{J'_x}\subseteq    U$.    This    topology    is   named    the
  \vip{finiteness topology} on $\EvF A$.
\end{definition}

The collection of $V_{J'}$ where ${J'}$ ranges over $\ort{\F(A)}$ is a
filter   basis.     Indeed,   for   every   $J'_1,J'_2\in\ort{\F(A)}$,
$V_{J'_1}\cap  V_{J'_2}=V_{J'_1\cup  J'_2}$  and  $J'_1\cup  J'_2  \in
\ort{(\F(A))}$.   Besides,  $\EvF A$  is  Hausdorff,  since for  every
$x\neq 0$  and $a\in\supp x$ the finite  set $\{a\}\in\ort{\F(A)}$, so
$x\notin V_{\{a\}}$.

Endowed  with  the  finiteness   topology,  $\EvF  A$  is  a  \vip{linearly
topologised space}. That is a  topological vector space over a discrete
field whose  topology is generated  by a fundamental system  (a filter
basis of  neighbourhoods of $0$,  here the $V_{J'}$, which  are linear
subspaces  of  $\EvF  A$).   Introduced  by  Lefschetz~\cite[II  -  \S
6]{lefschetz:at}, linearly topologised spaces have been widely studied
in~\cite[\S 10-13]{kothe:tvs1}.
\begin{definition}
  Let  us  call  $\LF$  the  category whose  objects  are  the  linear
  finiteness  spaces   and  whose  maps  are   the  linear  continuous
  functions.
\end{definition}
\begin{example}
  \emph{Booleans.}   As every  linear  finiteness space  whose web  is
  finite,  the linear  finiteness  space associated  with the  boolean
  relational     space      has     a     finite      dimension:     $
  \EvF{\bool}=\EvF{\ort\bool}=\Field^\bweb\simeq  \Field^2$. The space
  $\EvF{\bool}$ is endowed with the discrete topology since $\bool$ is
  antifinitary  and   so  $V_\bool=\{0\}$  is   a  fundamental  linear
  neighbourhood of $0$.
  \\
  \emph{Integers.}   The  linear   finiteness  space  associated  with
  $\mathcal{N}$  is the set  of finite  sequences over  $\Field$.  The
  linear finiteness  space associated with  $\ort{\mathcal{N}}$ is the
  set  of all  sequences over  $\Field$.  Since $\int\in\F(\ort\nat)$,
  $V_{\int}=\{0\}$  is a  neighbourhood of  zero, $\EvF\nat$ is
  endowed with the discrete topology.  On the contrary, the topology on
  $\EvF{\ort\nat}$  is  non-trivial:  the  fundamental system  is  the
  collection of  $V_{J'}$ where ${J'}$  ranges over finite  subsets of
  $\int$.  The space  $\EvF{\ort\nat}$ is simply $\Field^\int$ endowed
  with the usual product topology.
\end{example}

Linearly topologised  spaces are  quite different from  Banach spaces.
Any open  subset $V_{J'}$  of a finiteness  space is  closed ($\forall
x\notin V_{J'},\, (x+V_{J'})\cap V_{J'}=\emptyset$) --- linear
finiteness spaces  are totally disconnected.   Intuitively, unit balls
are replaced by subspaces, the $V_{J'}$s, which can hardly be considered
as bounded in the usual meaning.  However, there are linear
variants of the classical notions of boundedness and compactness:
\begin{definition}
  A subspace $C$  of $\EvF A$ is said  \vip{linearly bounded} iff for
  every ${J'}\in\ort{\F(A)}$ the codimension  of $V_{J'}\cap C$ in $C$
  is  finite, i.e.  there  exists a  subspace $C_0$  of $C$  such that:
\vspace{-1.5em}  \begin{center}
  $C=(V_{J'}\cap C) \oplus C_0$ and $\dim C_0$ is finite.
  \end{center}\vspace{-.5em}
  A subspace $K$ of $\EvF A$ is said \vip{linearly compact} iff for every
  filter $\F$ of affine closed subspaces of $\EvF A$, which satisfies
  the   intersection   property   (i.e.    $\forall   F\in\F,   F\cap
  K\neq\emptyset$), 
  \vspace{-1.5em}
  \begin{center} $(\cap\F)\cap K\neq\emptyset$.\end{center}\vspace{-.5em}
\end{definition}

\begin{theorem}[Tychonov]\cite[\S 10.9(7)]{kothe:tvs1}\label{th:tychonov}
  For  any  set $I$,  $\Field^I$  endowed  with  the product  topology
  (generated by $V_{J}=\{x\in\Field^I\st \supp x \cap {J}=\emptyset\}$
  with $J\subseteq I$) is linearly compact.
\end{theorem}
In the converse direction, we get a characterisation of
linearly compact spaces:
\begin{theorem}\cite[II - \S 6(32.1)]{lefschetz:at}
  \label{th:tychonovconverse}
  For every linearly compact vector space $K$, there is a set $I$ such
  that $K$ is topologically  isomorphic to $\Field^I$ endowed with the
  product topology.
\end{theorem}
\begin{example}
  \emph{Booleans.} As  in every  linearly topologised space  of finite
  dimension (see \cite[\S 13.1]{kothe:tvs1}),
  every subspace of $\EvF\bool$ is linearly bounded.\\
  \emph{Integers.}  It follows from Th.~\ref{th:tychonovconverse} that
  a linearly compact space is discrete iff its dimension is
  finite.  
  Hence, the  linearly compact subspaces of $\EvF\nat$  are the finite
  dimensional ones. Thanks to Tychonov Th.~\ref{th:tychonov},
  $\EvF{\ort\nat}$ is linearly compact.
\end{example}

In  the  finiteness setting,  finitary  support characterise  linearly
bounded spaces.   Although it is not  true in the general
setting of  linearly topologised spaces~\cite[\S 13.1(5)]{kothe:tvs1},
linearly compact  spaces are exactly the closed linearly bounded spaces.
\begin{proposition}\label{prop:lc,lb,fin}
\newcounter{lc,lb,fin}
\setcounter{lc,lb,fin}{\value{theorem}}
\addtocounter{lc,lb,fin}{-1}
Let $K$ be  a subspace of
  $\EvF A$.  There is an equivalence between
  \vspace{-0.5em}
  \begin{shortenumerate}
  \item\label{prop:lb} $K$ is linearly bounded,
  \item\label{prop:fin} $\supp  K=\cup\{\supp x\st x\in K\}$ is
    finitary,
  \item\label{prop:lc} the closure of $K$ is linearly compact.
  \end{shortenumerate}
\end{proposition}
\begin{proof}
  First, let $C$ be a linearly bounded space and ${J'}\in\ort{\F(A)}$.
  There  is a  finite  dimensional  subspace $C_0$  of  $C$ such  that
  $C=(C\cap  V_{J'})\oplus  C_0$.  Since  the  dimension  of $C_0$  is
  finite,   $\supp{C_0}$  is   finitary.   Besides,   $\supp   C  \cap
  {J'}=\supp{  C_0}\cap  {J'}$  which  is finite:  $\supp  C\in\F(A)$.
  Conversely, if  $\supp C\in\F(A)$, then  $\supp C\cap J'$  is finite
  and  $C\subseteq (\Field^{\supp K}\cap  V_{J'})\oplus \Field^{{\supp
      K}\cap    J'}$   is    linearly   bounded.     The   equivalence
  between~\eqref{prop:fin} and~\eqref{prop:lc} has already been proved
  in~\cite{ehrhard:fs-presheaves}.      \annex{Complete    proof    in
    Annex~\ref{ann:finit}, Prop.~\ref{ann:lc,lb,fin}}
\end{proof}

We focus attention on the  topological dual --- a linearly topologised
space endowed with the compact  open topology, that is the topology of
uniform  convergence on  either  linearly bounded  spaces or  linearly
compact spaces (equivalent thanks to Prop.~\ref{prop:lc,lb,fin}). 
\begin{definition}\label{def:dual}
  The \vip{topological  dual} $\EvF A '$  is the linear  space made of
  continuous linear forms over $\EvF  A$ and endowed with the linearly
  compact  open  topology.  This  topology  is  generated  by the 
  $
  \ann_{\EvF   A'}(K)=\{x'\in\EvF   A  '\st   \forall   x\in
  K,\,\<x',x\>=0\}$s
 where  $K$  ranges  over
  linearly compact subspaces of $\EvF A$.
\end{definition}

The two following propositions are central in 
the totality introduced in Section~\ref{section:total}.
\begin{proposition}[Separation]\cite[\S10.4(1')]{kothe:tvs1}.
  \label{lem:separation}
  \refannex{separation} For  every closed subspace  $D$ of $\EvF  A$ and
  $x\notin D$, there  is a continuous linear form  $x'\in\EvF A'$ such
  that $\<x,x'\>=1$ and  $\forall d\in D,\, \<d,x'\>=0$.  \annex{Proof
    in Annex~\ref{ann:finit}, Prop.~\ref{ann:separation}}
\end{proposition}

\begin{proposition}[Separation in the dual]
  \label{dualseparation} 
  \refannex{dualseparation}
  Let
  $T'$ be  a closed  affine subspace of  $\EvF A'$ such  that $0\notin
  T'$.   There   exists  $x\in  \EvF  A$  such   that  $\forall  x'\in
  T',\,\<x',x\>=1$.
\end{proposition}
\begin{proof}
  First, we  linear algebra ensures  the result when the  dimension of
  $T'$       is       finite       (c.f.        Annex~\ref{ann:finit},
  Lem.\ref{lem:finitecollection}).   Second, the closed  subspace $T'$
  does not contain  $0$, so there is $K$ linearly  compact of $\EvF A$
  such  that $\ann(K)\cap  T'=\emptyset$.   We use  the closed  affine
  filter made of  $T_{F'}=\{x\in E\st \forall x'\in F',\,\<x',x\>=1\}$
  with  $F'$ ranging  over finite  collections  of $\EvF  A'$ and  the
  compactness  of $K$  to  build the  wanted  $x$.  \annex{Details  in
    Annex~\ref{ann:finit}, Prop.~\ref{ann:dualseparation}}
\end{proof}

Both separations  theorem, ensure the algebraic  isomorphism between a
linear finiteness space and its second dual.  This isomorphism is both
continuous  and  open,  as  linearly  bounded subspaces  of  the  dual
coincide  with equicontinuous  subspaces  (Prop.~\ref{prop:equic}). To
sum  up, the  reflexivity of  finiteness  spaces relies  on the  links
between   linearly   compactness,    closed   linearly   bounded   and
equicontinuity.\annex{Annex~\ref{ann:finit},
  Prop~\ref{ann:refl}-\ref{ann:dual}}
\begin{proposition}(Equicontinuous spaces)
  \label{prop:equic}   \refannex{equic}  Let   $A$  be   a  relational
  finiteness space. A  subspace $B'$ of $\EvF A'$  is linearly bounded
  if and only  if there is $J'\in \ort{\F(A)}$  such that $B'\subseteq
  \ann_{\EvF  A'}(V_{J'})$.  \annex{proof  in  Annex~\ref{ann:finit},
    Prop.~\ref{ann:equic}}
\end{proposition}
Linear finiteness spaces satisfy other good properties (they admit
Schauder bases $(e_a)_{a_\in\supp A}$ and are
complete~\cite{ehrhard:fs}).
Although we do not know by now if the category of linearly topologised
spaces satisfying all these properties is stable under LL
constructions, we already know that the full subcategory of finiteness
spaces is  a model of  LL. 

\subsection{A model of MELL with MIX}
\label{section:model}

Both categories $\RF$ and $\LF$ constitute a model of classical linear
logic as  it has  been proved by  Ehrhard~\cite{ehrhard:fs}.  Although
linear finiteness  spaces are entirely determined  by their underlying
relational finiteness space~(Fig.\ref{fig:RF}),  we give the algebraic
description of the constructions  of LL in $\LF$~(Fig.\ref{fig:LF}) as
in~\cite{ehrhard:fs,ehrhard:fs-presheaves}.

\begin{figure}[ht!]
  \centering
  \begin{minipage}[c]{1\textwidth}
    \[
    \begin{array}{l@{\quad }|@{\quad }l}
      \ \text{\vip{Multiplicatives:}} & \ \text{\vip{Additives:}}\\
      \supp{A\parr B} = \supp{A\otimes B}=\supp A\times\supp B 
      & \supp{\with_i A_i}= \supp{\oplus_i A_i}=\textstyle\bigsqcup_i\supp{A_i}, \\\\
      \F(A\parr B)=\left\{
        \begin{array}{l}
          R\subseteq\supp A\times\supp B\text{ s.t.}\\
          \textstyle        \forall u\in\ort{\F(A)},\,\ R\cdot u\in\F(B)\\
          \textstyle       \forall v\in\ort{\F(B)},\,\trans R\cdot v\in\F(A)
        \end{array}
      \right\}
      &
      \F(\with_i A_i)=\left\{
        \begin{array}{l}
          \sqcup_i u_i\text{ s.t. }\\\forall i\in I,\,
          \textstyle  u_i\in\F(A_i)
        \end{array}
      \right\}
      \\\\
      \F(A\otimes B)=
      \left\{
        \begin{array}{l}
          R\subseteq\supp A\times\supp B\text{ s.t. }\\
          \quad\quad\ R\cdot \supp B\in \F(A)\\
          \quad\quad\trans R\cdot \supp A\in \F(B)
        \end{array}
      \right\}
      &
      \F(\oplus_i A_i)=
      \left\{
        \begin{array}{l}
          \sqcup_{j\in J} u_j\text{ s.t. }\\
          \multicolumn{1}{r}{\textstyle J\subseteq I\text{ finite }}\\
          \textstyle \forall j\in J,\, u_j\in\F(A_j)
        \end{array}
      \right\}
    \end{array}\]
    \hrule\vskip-.5em
    \begin{gather*}
      \text{\vip{Exponentials:}}\\
      \supp{\oc  A}=\supp{\wn  A}=\M{\supp  A}\ =\  \left\{\mu:\mathbb
        A\to \int\st
        \mu(a)>0 \text{ for finitely many }a\in A\right\}\\
      \F(\oc A)=\left\{M\subseteq\M{\supp
          A}\ \st\ \cup\{\supp{\mu},\,\mu\in M\}\in\F(A)\right\}\\
      \F(\wn      A)=\left\{M\subseteq\M{\supp      A}\st      \forall
        u\in\ort{\F(A)}, \M{u}\cap M \text{ finite}\right\}
    \end{gather*}
  \end{minipage}
  \caption{LL interpreted  in $\RF$; proofs are interpreted  as in the
    relational model. }\label{fig:RF}
\end{figure}


\vspace{-1em}
\begin{example}
  If  $\bool=1\oplus  1$,  then  we  get  back  to  Ex.~\ref{ex:bool}:
  $\sem{\bool}=\Field\oplus\Field\simeq \Field^2$.
  \begin{gather*}
     \supp{\wn \ort\bool}=\supp{\oc \bool}=\M{\true,\false}\simeq \int^2,\\
   \F(\oc\bool) =\left\{M\subseteq\M{\true,\false}\ |\ \cup_{\mu\in M
    }\supp{\mu}\in \F(\bool)\right\}= \Part{\int^2},\\
      \F(\wn\ort\bool)     =\left\{M\subseteq\int^2\    |\    \forall
    M'\subseteq\int^2,\,M\cap M'\text{ fin.}\right\}=\PartF{\int^2}.
  \end{gather*}
\end{example}
\vspace{-1em}

\begin{figure}[ht!]
  \begin{minipage}[c]{1\textwidth}
    \[\begin{array}{l@{\qquad\qquad}|@{\quad}l}
      \ \text{\vip{Multiplicatives:}} & \ \text{\vip{Additives:}}\\
      \EvF{\bot}=\EvF{1}= \Field & \EvF{\top}=\EvF{0}=\{0\}\\
      \EvF{A\parr B}= \EvF A \lparr \EvF B & \EvF{\with_{i\in I} A_i}=\times_{i\in I}\EvF A_i\\
      \EvF{A\otimes B}= \EvF A\ltens \EvF B &\EvF{\oplus_{i\in I} A_i}=\oplus_{i\in I}\EvF A_i\\
      \EvF{A\multimap B}= \Lc{\EvF A}{\EvF B} &
    \end{array}\]
    \hrule\vskip-1em
    \begin{gather*}
      \text{\vip{Exponentials:}}\\
      \EvF{\wn\ort    A}=\Polc{\EvF     A}    \qquad\qquad    \EvF{\oc
        A}=\left[\Polc{\EvF A}\right]'
    \end{gather*}
  \end{minipage}
  \caption{Interpretation of formul\ae\ in $\LF$, for proofs, see Annex~\ref{annex:interp}.}\label{fig:LF}
\end{figure}
\clearpage
\noindent Let us give some explanations on~Fig.\ref{fig:LF}:

{\it $(\multimap)$ continuous linear functions.}  We can
generalise  the topological  dual  framework and  endow  the space  of
continuous linear  functions $\Lc{\EvF  A}{\EvF B}$ with  the linearly
compact   open   topology.   It   is   generated  by   $\W(K,V)=\{f\st
f(K)\subseteq V\}$ where $K$ ranges over linearly compact subspaces of
$\EvF A$  and $V$ over fundamental  neighbourhoods of $0$  of $\EvF B$.\\
The linearly  topologised space  $\Lc{\EvF A}{\EvF B}$  coincides with
the  linear  finiteness  space  $\EvF{A\multimap  B}=\EvF{\ort  A\parr
  B}$. Indeed,  the canonical map  which maps each linear  function to
its matrix in the base induced by the web is a linear homeomorphism.

{\it $(\parr)$ hypocontinuous bilinear forms.}
As  noticed  by   Ehrhard,  the  evaluation  map:  $ev:\EvF{A\multimap
  B}\times  \EvF  A \to  \EvF  B$  is  separately continuous  but  not
continuous. That is why we need another the notion of hypocontinuity:\\
A   bilinear   form  $\phi:\EvF   A\times\EvF   B\to\Field$  is   said
\vip{hypocontinuous} iff for every  linearly compacts $K_A$ of
$\EvF A$ and $K_B$ of $\EvF  B$, there are two neighbourhoods $V_B$ of
$\EvF  B$  and $V_A$  of  $\EvF  A$  such that  $\phi(K_A,V_B)=0$  and
$\phi(V_A,K_B)=0$.   We denote  $\EvF A\lparr  \EvF B$,  the  space of
hypocontinuous bilinear  forms over $\EvF  A'\times\EvF B'$.  It  is a
linearly  topologised  space when  it  is  endowed  with the  linearly
compact          open          topology          generated          by
$\W(K'_A,K'_B)=\{\phi\st\phi(K'_A,K'_B)=0\}$ where $K'_{A}$ and $K'_B$
range  over linearly  compact subspaces  of  $\EvF A'$  and $\EvF  B'$
respectively.\\
The space  $\EvF A\lparr\EvF B$  is related to the inductive tensor
product~\cite{grot:tens} which was generalised to linearly topologised
space in~\cite{fischer:tens}.

{\it $(\tensor)$ complete tensor product.}
The  dual  $\EvF A\ltens\EvF  B$  of  $\EvF  A'\lparr\EvF B'$  is  the
completion  of the algebraic  tensor product  $\EvF A\otimes  \EvF B$.
Indeed, $\alpha(\EvF A\otimes\EvF B)$  is dense in $\EvF A\ltens\EvF
B$~\cite[Th 2.12]{fischer:tens} where
\[\begin{array}{rccc}
  \alpha:&  \EvF A\otimes \EvF B & \hookrightarrow &(\EvF A'\lparr \EvF B')'\\
 & x\otimes y &\mapsto& [x\ltens y:\,\phi\mapsto \phi(x,y)].
\end{array}\]
{\it $(\with)$ product.} Let $I$ be a set.
  The linear finiteness space $\EvF{\with_{i\in I} A_i}$ is the product of
  the $\EvF A_i$s, endowed with the product topology. 
  {\it   $(\oplus)$  direct   sum.}   The   linear   finiteness  space
  $\EvF{\oplus_{i\in  I} A_i}$ is  the coproduct  $\oplus_{i\in I}\EvF
  A_i$ (made  of finite linear  combinations of elements of  the $\EvF
  A_i$s), endowed with the topology induced by the product topology.





{\it $(\oc)$ through webs.}
The  comonadic  structure   $(\EvF{\oc  A},\epsilon,\delta)$  and  its
\vip{linear distribution}  $\kappa$, can be described  with respect to
the web base: for $x=\sum_{a\in\supp A} x_a e_A$ given in $\EvF A$, we
set   $x^\mu=\prod_{a\in\supp   a}x_a^{\mu(a)}$   and  we   take   $X=
\sum_{\mu\in\M{\supp A}} x_\mu e_\mu\in \EvF{\oc A}$ in
\[\begin{array}{c@{\quad\quad}c}
  \kappa :\  x\mapsto \sum_{\mu\in\M{\supp
      A}} x^\mu e_\mu &
  \epsilon:\    X  \mapsto
  \sum_{a\in\supp A}
  x_{[a]}e_a, \\\\
  \multicolumn{2}{c}{
  \delta   :\  X  \mapsto
  \sum_{M\in\M{\M{\supp A}}}\ \left(\sum_{\mu=\Sigma(M)} x_\mu\right)\ e_M.}
\end{array}\]
The exponentiation $\exp x=\kappa(x)$ of $x$ 
satisfies $\epsilon(\exp x)=x$ and $\delta(\exp x)=\exp{(\exp x)}$.

{\it  $(\wn)$  analytic   functions.}   The  linear  finiteness  space
$\EvF{\wn \ort A}$ is the dual  of $\EvF{\oc A}$.  However, there is a
more  algebraic  approach~\cite{ehrhard:fs-presheaves}  of the monoid
$\EvF{\wn \ort A}$.  A function $P$ is \vip{polynomial} whenever there
are  symmetric hypocontinous  $i$-linear forms\footnote{Hypocontinuity
  for  $i$-linear forms  is a   generalisation  of the
  bilinear           case.\annex{Annex~\ref{annex:interp},
    Def.~\ref{def:hypo}}} $\phi_i:\times^i\EvF{A}'\to\Field$ such that
$P(x)=\sum_{i=0}^n\phi_i(x,\dots,x)$.   The space  $\EvF{\wn  \ort A}$
coincides with  the completion of  polynomial functions over  $\EvF A$
endowed with the  linearly compact open topology\footnote{generated by
  $\W(K)=\{P\text{ polynomial function}\st P(K)=0\}$ with $K$ linearly
  compact.}.  We call the elements of $\EvF{\wn\ort A}$ \emph{analytic
  functions}.

{\it  $(\oc)$  distributions.}  Finally,  we  are  concerned with  the
\emph{Taylor expansion}  formula of Ehrhard~\cite{ehrhard:fs}.  Taking
into account that $\EvF{\oc A}$  is the dual space of $\Polc{\EvF A}$,
we can  establish a parallel with distributions.   For instance, $\exp
x$ sends  an analytic function  $F$ to its image  $\<\exp x,F\>=F(x)$,
hence it  corresponds to  the \emph{dirac mass}  at $x$.   Besides, in
$\LF$, there exists a sequence of \emph{projections}:
\vskip-1em
\begin{equation*}
\textstyle
\pi_n:  \sum_{\mu\in\M{\supp   A}}x_\mu  e_\mu\in  \EvF{\oc  A}\mapsto
\sum_{\sharp\mu=n}x_\mu e_\mu\in\EvF{\oc A},
\end{equation*}\vskip-0.4em
\noindent which are linear and
continuous   since   their   support   $\supp\pi_n=\left\{(\mu,\mu)\st
  \sharp\mu=n\right\}$   are  finitary.  The   vector  $x^n=\pi_n(\exp
x)=\sum_{\sharp \mu=n}x^\mu e_\mu$ of $\EvF{\oc A}$ is the convolution
of  $x$  iterated  $n$  times.  This distribution  sends  an  analytic
function  to a  homogeneous  polynomial  of degree  $n$,  that is  its
derivative at zero.  From $\exp x=\sum_{n=0}^\infty \frac{1}{n!} x^n$,
Ehrhard deduces the Taylor expansion formula:
\vskip-1.5em
\begin{equation}
   \label{eq:taylor}
 F(x)=\sum_{n=0}^\infty\frac{1}{n!}\< x^n,F\>.  
 \end{equation}

\begin{example}
$    \EvF{\oc \bool}
    =    \left\{z\in\Field^{\int^2}\st \supp z\in\Part{ \int^2}\right\}
    = \Ser{X_t,X_f}$\\
$    \EvF{\wn \ort\bool}
    = \left\{z\in \Field^{\int^2}\st \supp z\in\PartF{ \int^2}\right\} 
    = \Pol{X_t,X_f}$\\
$    {\EvF{\oc           \bool\multimap          \bool}=
      \EvF{\oc\bool\multimap      1\oplus      1}=\EvF{\wn\ort\bool}^2
      =\Pol{X_t,X_f}\times \Pol{X_t,X_f},}$\\
$    \EvF{\parr^n\wn\ort\bool}=\Pol{X_1,X_2,\dots,X_{2n-1},X_{2n}}$\\
$    \EvF{\otimes^n\oc\bool\multimap\bool}=\Pol{X_1,X_2,\dots,X_{2n-1},X_{2n}}^2$
\end{example}

\section{Totality and Barycentric lambda-calculus}
\label{section:total}

In the  present section, we  explore an algebraic version  of totality
spaces, where formul\ae\ are  interpreted as finiteness spaces with an
additional totality  structure.  Adapting Loader's  definition to this
algebraic  setting,  we define  a  general  concept of  \emph{totality
  finiteness space}: it is a  pair $\pair{\EvF A}\T$ where $\EvF A$ is
a linear  finiteness space and $\T$ is  a subset of $\EvF  A$ which is
equal  to  its   second  dual  for  a  duality   associated  with  the
\emph{polar} as defined  below. Actually, the finiteness space
interpreting  any formula coincides  with the  first component  of the
totality finiteness space interpreting this formula.

\noindent{\bf Notations:}
\begin{runitemize}
\item[$\ldir(T)$] is the direction of the affine
  space $T$ and if $x\in T$, $T=x+\ldir(T)$,
\item[$\adj f:$]$y^\ast\in F^\ast\mapsto [x^\ast\in E^\ast:x\mapsto
  \<y^\ast,f(x)\>]$ is the linear adjoint of 
  $f:E\to F$.
\item[$A$ and $B$] are relational finiteness spaces.
\end{runitemize}

\subsection{Totality finiteness spaces.}
The polar orthogonality is defined as follows:
\begin{equation*}
  \forall x\in\EvF A,\, x'\in\EvF A',\ x\perp^\bullet x'\iff \<x',x\>=1
\end{equation*}
The polar of a subset $\T$ of $\EvF A$ is the following closed affine subspace of
$\EvF A'$:
\begin{equation*}
  \pol \T=\{x'\in\EvF A'\st \forall x\in\T,\,\<x',x\>=1\}
\end{equation*}
This  set  is closed  since  $(x,x')\mapsto  \<x',x\>$  is linear  and
separately continuous on $\EvF{A}\times \EvF{A}'$ (let
$(x,x')\in  \EvF{A}\times \EvF{A}'$, $\{x\}$  is linearly  compact, so
$\ann(x)$ is open  in $\EvF A'$; $x'$ is a  linear continuous form, so
$\ker(x')$   is  open  in   $\EvF{A}$).   Notice   that,  up   to  the
homeomorphism between $\EvF{A}$ and $\EvF{A}''$, if $\T'$ is an affine
subspace of $\EvF{A}'$, $\pol{\T'}=\{x\in \EvF{A}\st\forall x'\in
\T',\,\<x',x\>=1\}$.

There is  a simple characterisation  of polar-closed  affine subspaces:
\begin{proposition}[Characterisation]\label{prop:charac}
  A subspace $\T$ of  $\EvF A$ is polar-closed iff it is the
  empty set, the space $\EvF A$, or a closed affine subspace that does
  not contain $0$.
\end{proposition}
\begin{proof}
  If       $\T=\EvF{A}$,      then       $\pol{\T}=\emptyset$      and
  $\pol{\emptyset}=\EvF{A}=\T$.       If      $\T=\emptyset$,     then
  $\dpol{\T}=\emptyset$.  It  remains the  case where $\T$  is affine,
  closed  and $0\notin\T$.  The  inclusion $\T\subseteq  \dpol{\T}$ is
  straightforward.  Let  us prove the  contrapositive.  Let $x_0\notin
  \T$.   Let $z_0\in\T$  and $D=\ldir(\T)$  then  $\T=z_0+D$, $x_0\neq
  z_0$      and       $x_0-z_0\notin      D$.       By      separation
  Prop.~\ref{lem:separation},  there is  $x'_0\in \EvF{A}'$  such that
  $\<x'_0,x_0-z_0\>=1$  and $\forall  d\in D,\,\<x'_0,d\>=0$.   On the
  one     side,     if     $\lambda=\<x'_0,z_0\>\neq0$,     we     set
  $y'_0=\frac{1}{\lambda}x'_0$,  then $\<y'_0,  z_0\>=1$  and $\forall
  d\in  D,\   \<y'_0,d\>=\frac{1}{\lambda}\<x'_0,d\>=0$,  so  $y'_0\in
  \pol{\T}$.    However  $\<y'_0,x_0\>=\frac{1}{\lambda}  \<x_0,x'_0\>
  =\frac{1+\lambda}{\lambda}\neq 1$,  hence $x_0\notin \dpol{\T}$.  On
  the  other  side,  $\<x'_0,z_0\>=0$, then  $\<x'_0,x_0\>=1$.   Since
  $0\notin  \T$ and  by  separation Prop.~\ref{lem:separation},  there
  exists $x'_1\in  \EvF{A}' $ such that  $\<x'_1,z_0\>=1$ and $\forall
  d\in  D,\,  \<x'_1,d\>=0$,   hence  $x'_1\in  \pol{\T}$.   Moreover,
  $\<x'_1+x'_0,z_0\>=1$  and  $\forall  d\in  D,\,  \<x'_1+x'_0,d\>=0$
  hence  $x'_1+x'_0\in\pol\T$.  To  conclude,  either $\<x'_1,x_0\>=0$
  and  $x'_1\in\pol{\T}$  or $\<x'_0+x'_1,x_0\>=1+\<x_0,x'_1\>\neq  1$
  and $x'_0+x'_1\in\pol{\T}$, so in both cases, $x_0\notin\dpol\T$.

\qed
\end{proof}

From  this  characterisation, we  deduce  another  one  which will  be
useful to compute the constructions of the model.
\begin{corollary}\label{cor:aff}\refannex{aff}
  Let $\T$  be a subset  of $\EvF A$.  If  $\pol\T\neq\emptyset$, then
  $\dpol\T=\caff(\T)$.
\end{corollary}
\begin{proof}
  The proof is based on:  $\T\subseteq\caff(\T)\subseteq\dpol\T$.   
  \annex{Proof in Annex~\ref{annex:totalite}}
\end{proof}

\begin{definition} 
  A  \vip{totality  finiteness   space}  is  a  pair  $\pair{\EvF
    A}{\T}$ made of a  linear finiteness space $\EvF{A}$ and a
  \vip{totality candidate}  $\T$, that is is a  polar closed subspace
  of $\EvF{A}$.

  Let  $\TF$ be  the category  whose objects  are  totality finiteness
  spaces  and whose  morphisms  are continuous  linear functions  that
  preserve the totality candidates.
\end{definition}

\subsection{A model of classical linear logic}
\label{totality}

To prove that  $\TF$ is a model of classical linear  logic, we use the
definitions and results  of~\cite[\S 4-5]{hyland:goll}.  Let $\GLF$ be
the  double glueing  of the  category $\LF$  along  the $\mathrm{HOM}$
functor.  The  objects of $\GLF$ are  triples $\triple{\EvF A}{U}{U'}$
where $U$  and $U'$ are subspaces  of respectively $\EvF  A$ and $\EvF
A'$.   A morphism between  $\triple{\EvF A}{U}{U'}$  and $\triple{\EvF
  B}{V}{V'}$ is a continuous  linear function $f:\EvF A\to\EvF B$ such
that $f(U)\subseteq V$ and $\adj f(V')\subseteq U'$, where $\adj
f$ 
is the  adjoint of  $f$.  The linear  exponential comonad of  $\LF$ is
equipped  with  a  well-behaved linear  distribution  $\kappa:x\in\EvF
A\mapsto \exp  x\in\EvF{\oc A}$ (it  is routine to check  the diagrams
satisfied
by $\kappa$).\\
The category $\TF$ is a  subcategory of $\GLF$
(considering triples  $\triple{\EvF  A}{\T}{\pol\T}$).  More
precisely,  it is  a \emph{Tight  orthogonality} with  respect  to the
polar orthogonality.  This orthogonality  is \emph{stable} since it is
\emph{focussed} with  respect to  the focus $\{1\}$:  $ x\perp^\bullet
x'\iff \<x',x\>=1\iff  x'(x)=1$.  Since $\LF$ is a  model of classical
linear  logic,  $\TF$  is  also  a model  of  classical  linear  logic
(c.f. \cite[Th.~$5.14$]{hyland:goll}).

The  constructions   inherited  from  $\LF$   as  described
in~\cite[\S 5.3]{hyland:goll} are:
\[\begin{array}{rcl@{\qquad}rcl}
  \multicolumn{6}{c}{\T(\ort  A)=\pol{\T(A)}}\\
 \multicolumn{3}{c}{\T(1)=\T(\bot)=\{1\},}& \multicolumn{3}{c}{\T(0)=\T(\top)= \{0\},}\\
 \T(A\otimes B)&=& \dpol{[\T(A)\otimes\T(B)]}, &
  \T(A\with B)&=&\T(A)\times\T(B),\\
\T(A\multimap B) &=& \pol{[\T(A)\otimes \pol{\T(B)}]},& \
T(A\oplus B)&=&\pol{[\T(A)\times\T(B)]},\\
\multicolumn{6}{c}{
  \T(\oc A)=\dpol{[\kappa(\T(A))]}=\dpol{\left\{\exp x\st x\in\T(A)\right\}}}
\end{array}
\]
Moreover,  we describe  every totality  candidate as  a  closed affine
subspace. This  algebraic description is  made possible thanks  to the
characterisation  of totality candidate  (Prop.~\ref{prop:charac}) and
to the algebraic setting.
\begin{proposition}
  \refannex{mult}
  \begin{eqnarray}
    \T(A\otimes       B)&=&\caff(\T(A)\otimes\T(B)),\nonumber\\
    \T(A\multimap B) &=& \left\{f\in\EvF A\st f(\T(A))\in\T(B)\right\}.\label{eq:fun}\\
    \T(A\oplus
    B)&=&\caff{(\T(A)\times\ker(\pol{\T(B)})\cup\ker(\pol{\T(A)})\times
      \T(B))}\\
    \T(\oc A)&=&\caff(\exp x\st x\in\T(A))\label{eq:exp},\\
    \T(\wn    A)&=&\left\{F\in\Polc{\EvF    A}\st   \forall
      x\in\T(A),\,F(x)=1\right\}\label{eq:wn},\\
    \T(\oc A\multimap B)&=&\left\{F\in\Polc{\EvF A,B}\st \forall
      x\in\T(A),\,F(x)\in\T(B)\right\}\label{eq:kleisli}.
  \end{eqnarray}
\end{proposition}
\begin{proof} The proof  relies on showing that $\pol\T$  is not empty
  and to use Cor.~\ref{cor:aff}.

\annex{
Details in Annex~\ref{annex:totalite}, Prop~\ref{ann:mult}-\ref{ann:exp}.}
\end{proof}

The formula $A\Rightarrow B$ is interpreted as the totality finiteness
space that is made
of  morphisms  of  the   Kleisli  category.   The  totality  candidate
associated with  $A\Rightarrow B$ satisfies  the following fundamental
equation:
\begin{eqnarray*}
  \T(A\Rightarrow B)&=& \{F:\EvF A\to \EvF B\text{ analytic}\st \forall x\in
  \T(A),\,  F(x)\in \T(B)\}
\end{eqnarray*}
In other word, the totality we have defined is a logical relation.

\begin{example}\label{ex:Boolean}
  $    \Tot{\bool}=\{(x_t,y_t)\in\Field^2\st x_t+y_t=1\},\qquad
    \Tot{\ort\bool}=\{(1,1)\},$
  \[\begin{array}{l}
    \Tot{\oc \bool}=\{F\in\Ser{X_t,X_f}\st \forall (x_t,
    y_t),x_t+y_t=1\Rightarrow F(x_t, y_t)=1\},\\
    \Tot{\wn \ort\bool}=\{P\in\Pol{X_t,X_f}\st
    x_t+y_t=1\Rightarrow P(x_t, y_t)=1\},\\
    \Tot{\oc    \bool\multimap    \bool}=\{(P,Q)\in\Pol{X_t,X_f}^2\st
    x_t+y_t=1 \Rightarrow P(x_t, y_t)+Q(x_t,y_t)=1\}\\
    \multicolumn{1}{l}{
      \Tot{\parr^n\wn\ort\bool}=\{P\in\Pol{X_1,\dots, X_{2n}}\st \forall
      1\leq i\leq
      n,x_{2i-1}+x_{2i}=1}\\
    \multicolumn{1}{r}{
      \Rightarrow P(x_1, y_2,\dots,x_{2n-1},x_{2n})=1\}}\\
    \Tot{\otimes^n\oc\bool\multimap\bool}=\{(P,Q)\in\Pol{X_1,\dots,X_{2n}}^2\st
    P+Q-1 \in \Tot{\parr^n\wn\bool}\}
  \end{array}\]
\end{example}

Because totality candidate  are affine spaces, it is  natural to add a
barycentric construction to our proof  system and to interpret it by a
barycentric combination. Totality finiteness spaces constitute a model
of linear logic with MIX and barycentric sums.

\subsection{Simply typed boolean barycentric lambda-calculus}
\label{def:BarycentricCalculus}

We propose a $\lambda$-calculus in the style of Vaux's algebraic
$\lambda$-calculus~\cite{vaux:alglam2}.      In     the    barycentric
$\lambda$-calculus, sums of terms are  allowed.  It is well known that
the application  in $\lambda$-calculus is  linear in the  function but
not in its argument.  That is why we introduce two kinds of terms:
\emph{atomic terms} that  do not contain barycentric sums  but in the
argument  of an  application  and \emph{barycentric  terms} which  are
barycentric  sums of  atomic terms.  Moreover, we  add booleans  and a
conditional construction.

\subsubsection{Syntax.}

Let $\V$  be a countable set  of variables.
Atomic terms $\bb s$ and barycentric terms $\bb T$ are inductively defined
by
\begin{equation*}
  \begin{array}{rcl}
    \bb R,\bb S&::=& \sum_{i=1}^m a_i\,\bb s_i\qquad \text{ and } \left\{
      \begin{array}{l}
        \forall i\in\{1,\dots,m\}, a_i\in\Field\,,\\
        \sum_{i=1}^m a_i=1\,,
      \end{array}\right.  \\
    \bb s&::=& x\st \lambda x.\bb s  \st (\bb s)\bb T \st \true \st\false
    \st \ITE {\bb s}{\bb S} {\bb R}\qquad \text{ where }  \quad x\in\V 
  \end{array}
\end{equation*}
We  denote $\Lat$  the  collection  of atomic  terms  and $\Laff$  the
collection of barycentric terms.
We  quotient  all  these  sets  of terms  by  $\alpha$-conversion  and
associativity and commutativity of the sum.

\subsubsection{Types.}

The barycentric $\lambda$-calculus is simply typed with the usual type
system built  on $A$  and $A\Rightarrow B$  with the  restriction that
barycentric sums of atomic terms  are possible only if the latter have
the  same type.   Decomposing  $A\Rightarrow B$  with exponential  and
linear  map: $!A\multimap  B$,  the type  system  can be  reformulated
within linear logic:

\newcommand{\infer}[2]{\begin{array}{c} #1 \\ \hline #2 \end{array}} 
\[
\begin{array}{c @{\qquad\qquad} c}
\infer{x\in\V}{\Gamma ,x:A \vdash x:A}\quad(\mathrm{var})
&
\infer{x:A\in\Gamma\vdash  \bb   s  :B}{\Gamma\vdash   \lambda  x.\bb
 s:A\Rightarrow B}\quad (\mathrm{abs})
\\
\infer{\Gamma\vdash \bb  s\quad\Gamma\vdash \bb T}{\Gamma  \vdash (\bb
  s)\bb T}\quad (\mathrm{app})
&
\infer{\Gamma\vdash    \bb   s_i:A\quad    \sum\limits_{i=1}^n    a_i   \bb
  =1}{\Gamma\vdash \sum\limits_{i=1}^n a_i\, \bb s_i:A}\quad(\mathrm{sum})
\\
\infer{}{\Gamma\vdash\true:\bool}\quad(\mathrm{true})
&
\infer{}{\Gamma\vdash\false:\bool}\quad(\mathrm{false})
\\
\multicolumn{2}{c}{
  \infer{\Gamma\vdash   \bb   s:\bool\quad\Gamma\vdash\bb   S:\bool\quad
    \Gamma\vdash\bb R:\bool}{\Gamma\vdash \ITE{\bb s}{\bb S}{\bb R}}\quad(\mathrm{cond})}
\end{array}\]

\subsubsection{Semantics}

We interpret the barycentric $\lambda$-calculus in $\LF$ through a standard
translation of the $\lambda$-calculus in LL, extended to deal with the
barycentric and boolean features as follows:
\[\begin{array}{c@{\quad\qquad}c@{\qquad\quad}c}
\sem{\sum_{i=1}^n  a_i\,\bb s_i}^\Gamma  =\sum_{i=1}^n  a_i\, \sem{\bb
  s_i}^\Gamma&
 \sem{\true}^\Gamma=(1,0),&  \sem{\false}^\Gamma=(0,1),

\\\\ 
\multicolumn{3}{c}{   \sem{\ITE {\bb R} {\bb S} {\bb T}
  }^\Gamma=(\sem{\bb R}_t^\Gamma\sem{\bb S}_t^\Gamma+\sem{\bb R}_f^\Gamma\sem{\bb T}_t^\Gamma,
  \sem{\bb R}_t\sem{\bb S}_f^\Gamma+\sem{\bb R}_f^\Gamma\sem{\bb T}_f^\Gamma).
}
\end{array}\]
Notice
that since $\sem{\bool}=\Field^2$, the  semantics of each term $\bb S$
of type $\bool$ is given  by its two components $\sem{\bb S}=(\sem{\bb
  S}_t,\sem{\bb S}_f)$.

\begin{theorem}
  Totality  finiteness spaces  constitute a denotational  model of
  the barycentric $\lambda$-calculus.
\end{theorem}
Hence, the notion of totality we have defined is in line with the
notion  of  realisability  in   Logic  where  a  term  $\lambda  x.\bb
t:A\Rightarrow B$ is realisable iff  $\forall \bb s:A,\,\bb t[x\leftarrow \bb s]:B$.

\section{Towards completeness}
\label{section:comp}


We  focus  attention  on  closed  terms  of  type  $\bool^n\Rightarrow
\bool$. As we have  seen in Ex.~\ref{ex:Boolean}, terms of that
type are pairs of polynomials  $\mathbb
P=(\mathbb P_t,\mathbb  P_f)\in\Pol{X_1,\dots,X_{2n}}^2$ s.t.  for all
$(a_i)\in\Field^{2n}$      with      $a_{2i-1}+a_{2i}=1$,     $\mathbb
P_t(a_1,\dots, a_{2n})+\mathbb P_f(a_1,\dots,a_{2n})=1$. 

\begin{theorem}[Completeness]\label{th:PolBoolean}
  Every  total  function  of  $\T(\bool^n\Rightarrow  \bool)$  is  the
  interpretation of a term of the boolean barycentric calculus.
\end{theorem}
More  precisely, we  prove  that every  pair  of polynomials  $\mathbb
P\in\Tot{\otimes^n\oc\bool\multimap\bool}$      is     \emph{boolean},
i.e. there is  a term $\bb S$ of the boolean  calculus such that $\sem
{\bb S}=(\mathbb P_t,\mathbb P_f)$.

Let us  first introduce some  notations and intermediate  results.  
\begin{equation*}
  \begin{array}{rcl@{\qquad\quad}rcl}
    \lnot \bb S&=& \ITE{\bb S}{\false}{\true},
    &\sem{\lnot   \bb   S}&=&(\sem{\bb S}_f,\sem{\bb S}_t), \\
    {\bb  S}^+   &=&  \ITE{\bb  S}{\true}{\true},&   \sem{{\bb  S}^+}&=&
    (\sem{\bb S}_t+\sem{\bb S}_f,0),\\
    {\bb  S}^-   &=&  \ITE{\bb  S}{\false}{\false},&\sem{{\bb  S}^-}&=&
    (0,\sem{\bb S}_t+\sem{\bb S}_f),\\
    \bb\Pi_i  &=& \lambda  \bb  x_1,\dots,\bb x_n\cdot\bb  x_i,&\sem{\bb
      \Pi_i}&=&(X_{2i-1},X_{2i}).
  \end{array}
\end{equation*}

The following pairs of polynomials are boolean:
\begin{gather}
  (X_{2i},X_{2i-1})  = X_{2i}\cdot (1,0)+ X_{2i-1}\cdot (0,1) =\sem{\lnot {\bb\Pi}_i},\\
  (X_{2i-1}+X_{2i},0)  = X_{2i}\cdot (1,0)+ X_{2i-1}\cdot (1,0)= \sem{{\bb\Pi}_i^+},\label{eq:sum}\\
  (1-X_{2i},X_{2i}) = (1,0)+(X_{2i-1},X_{2i})-(X_{2i-1}+X_{2i},0)= 
  \sem{\true+{\bb\Pi}_i-{\bb\Pi}_i^+},\nonumber\label{eq:paire}\\
  (1-X_{2i-1},X_{2i-1})  
  = \sem{\true+\lnot {\bb\Pi}_i-{\bb\Pi}_i^+}.\label{eq:impaire}\nonumber
\end{gather}

We  prove first  a weak  version of  the completeness  theorem  where we
assume     that   $\mathbb    P_t+\mathbb     P_f-1$    vanishes
\emph{everywhere}.
\begin{lemma}[Affine pairs]\label{lem:PolProj}
  For   every  polynomial   $P\in\Pol{X_1,\dots,X_n}$,  the   pair  of
  polynomials $(1-P,P)$ is boolean.
\end{lemma}
\begin{proof}
  We  use an  induction on  the degree  $d$ of  $P$.  If  $d=0$, there
  exists      $a\in\Field$       such      that      $P=a$,      hence
  $(1-P,P)=\sem{a\,\true+(1-a)\,\false}$.
    
  If $d>0$, let  us first study the monomial  case, i.e.  $X^\mu=\prod
  X_i^{\mu_i}$ with, say, $\mu_1\geq 1$.
  \begin{eqnarray*}
    \textstyle (1-X^\mu,X^\mu)&=&\textstyle(1-X_1)\cdot(1,0)+ X_1\cdot\left(1-X_1^{\mu_1-1}\prod_{i\neq
        1}X_i^{\mu_i},X_1^{\mu_1-1}\prod_{i\neq 1}X_i^{\mu_i}\right)\\
    &=&\textstyle\sem{ \ITE{\bb \Xi_1}{\true}{\bb
        \Xi_{d-1}}}=\sem{\bb\Xi_\mu}.
  \end{eqnarray*}
  where  the induction  hypothesis ensures the  existence of
  $\bb\Xi_1$   and    $\bb\Xi_{d-1}$   respectively   interpreted   by
  $X_1^{\mu_1-1}$ and $\prod_{i\neq 1}X_i^{\mu_i}$.
  Finally, if $P=\sum a_\mu\prod X_i^{\mu_i}$, then
  \begin{eqnarray*}\textstyle
    (1-P,P)&=&\textstyle\left(1-\sum a_\mu\right)(1,0)+(\sum  a_\mu)
    \left(1-X^\mu, X^\mu\right)\\
    &=& \textstyle \sem{\left(1-\sum a_\mu\right)\true +(\sum
      a_\mu)\, \bb\Xi_\mu}.
  \end{eqnarray*}
\qed
\end{proof}

The following algebraic lemma, allows us to reduce our problem to \emph{Affine
pairs}.
\begin{lemma}[Spanning polynomials]\label{lem:PolVanish}
  Let  $P\in\Pol{X_1,\dots,X_{2n}}$  where  $\Field$  is  an  infinite
  field. If $P$ vanishes  on the common zeroes of $X_{2i-1}+X_{2i}-1$,
  then for every $i$ in $\{1,\dots,n\}$ there is
  $Q_i\in\Pol{X_1,\dots,X_{2n}}$ such that
$    P=\sum_{i=1}^n Q_i(X_{2i-1}+X_{2i}-1)$.
\end{lemma}
\begin{proof}
  Under the change of variable: 
  $\begin{array}{@{\quad}l@{\quad} l@{\quad}}
        Y_i=X_{2i-1}+X_{2i}-1,&
        Y_{i+n}= X_{2i},
      \end{array}
      $ for  $ i\in\{1,\dots,n\}$, we  denote by $P_Y$  the polynomial
      $P$.        Then       for       every       $(y_i)\in\Field^n$,
      $P_Y(0,\dots,0,y_{n+1},\dots,y_{2n})=0$.                    Since
      $\Pol{Y_2,\dots,Y_{2n}}$           is           a          ring,
      $\Pol{Y_2,\dots,Y_{2n}}[Y_1]$   is   an   euclidean  ring.   The
      euclidean  division of  $P_Y$ by  $Y_1$ gives  $P_Y=Q_1 Y_1+R_1$
      where           $R_1\in\Pol{Y_2,\dots,Y_{2n}}[Y_1]$          and
      $R_1'\in\Pol{Y_2,\dots,Y_{2n}}$.   By iterating this  process on
      $R_i$ for $i\in\{1\dots  n-1\}$, we get $P_Y=\sum_{i=1}^n Q_iY_i
      +R_n$ where $R_i\in\Pol{Y_1,\dots,Y_{2n}}$ and
      $R_n\in\Pol{Y_{n+1},\dots,Y_{2n}}$. For all $(y_i)\in\Field^n$, we have\\
      $P_Y(0,\dots,0,y_{n+1},\dots,y_{2n})=R_n(y_{n+1},\dots,y_{2n})=0$.
      Since  $\Field$ is infinite,  $R_n=0$ and  $P_Y=\sum_{i=1}^n Q_i
      Y_i$.   Change  the  variables,  we get,  $P=\sum_{i=1}^n  Q_{i}
      (X_{2i-1}+X_{2i}-1)$.

\qed
\end{proof}

\begin{proof}[Theorem~\ref{th:PolBoolean}]
  Let $\mathbb P\in \Tot{\otimes^n\oc\bool\multimap\bool}$.  Thanks to
  Ex.~\ref{ex:Boolean}, we  know that $\mathbb  P_t+\mathbb P_f-1$
  vanishes on every zero  of $\{X_{2i-1}+X_{2i}-1\st 1\leq i\leq n\}$.
  Then,  we  can  apply  Lem.~\ref{lem:PolVanish}: $ \mathbb P_t+\mathbb
  P_f-1=\sum_{i=1}^n Q_i (X_{2i-1}+X_{2i}-1)$ with 
  $Q_i\in\Pol{X_1,\dots,X_{2n}}$.  Thus,\\
$    (\mathbb P_t,\mathbb P_f)=\sum_{i=1}^n
    \left[(1-Q_i)\cdot(1,0)+Q_i\cdot(X_{2i-1}+X_{2i},0)\right]+(1-\mathbb
    P_f,\mathbb P_f)-n(1,0).$
  By Lem.~\ref{lem:PolProj},  there are  boolean terms $\bb  S_i$ and
  $\bb  T$  such  that  $(1-Q_i,Q_i)=\sem{\bb  S_i}$  and  $(1-\mathbb
  P_f,\mathbb P_f)=\sem{\bb  T}$.  We have  seen in Eq.~\eqref{eq:sum}
  that  $ (X_{2i-1}+X_{2i},0)=  \sem{\bb\Pi_i^+}$.   Finally, we  have
  found a term whose semantics is
  \begin{equation*}
    \mathbb         P=\sem{\sum_{i=1}^n         \left(        \ITE{\bb
          S_i}{\true}{\bb\Pi_i^+}\right) +\bb T -
      n\, \true}.
  \end{equation*}
\vskip-1.5em
\qed
\end{proof}



\begin{example}[Gustave and Por functions] Several pairs of
  polynomials can interpret the functions {$\por\in \T(\oc\bool\otimes\oc\bool)\multimap \bool$} and
  $\gus\in\T(\oc\bool\otimes\oc\bool\otimes\oc\bool)\multimap   \bool$
  satisfying:\vskip-.5em
  \[
  \begin{array}{l@{\qquad\qquad}l}
    \por(\true,0)=\true &\gus(\true,0,0)=\true\\
    \por(0,\true)=\true &\gus(0,\true,0)=\true\\
    \por(\false,\false)=\false &\gus(0,0,\true)=\true\\
    & \gus(\false,\false,\false)=\false
  \end{array}
  \]\vskip-.5em
  The pairs of polynomials with the smallest degree are respectively:\vskip-1.5em 
   \begin{align*}
    &\begin{array}{rccccl}
    \por:&\bool&\times&\bool&\Rightarrow& \bool\\
    &\left(x\right.&,&\left.y\right)           &\mapsto          &
    (x_t+y_f-x_ty_t,x_fy_f)
    \end{array} \\
   & \begin{array}{rccccccl}
    \gus:&\bool&\times&\bool&\times&\bool&\Rightarrow& \bool\\
    &\left(x\right.&,&y&,&\left.z\right)           &\mapsto          &
    (x_ty_f+y_tz_f+z_tx_f,x_ty_tz_t+x_fy_fz_f)
    \end{array}
  \end{align*}
\end{example}

\section*{Conclusion}
The first two sections of this article emphasise the algebraic and
topological description of the model of finiteness spaces. It is
important for three reasons. \\
First, the definition  of linear finiteness spaces is  so web oriented
that  the corresponding category  is not  obviously closed  by certain
operations such as quotients  or even subspaces. 
The purpose of a more algebraic approach is to get rid of bases. Our
description of reflexivity is a first step in this direction.\\
Second, our study  has unveiled an algebraic approach  to totality where
totality   candidates  admit  a   simple  algebraic   and  topological
characterisation;  such  a   characterisation  was  not  available  in
coherence spaces.  Moreover, although we needed to use linear logic to
describe algebraic totality,  we get a notion which  coincides with the
standard  intuitionist  hierarchy.   Notice  that the  non-stable  $\por$
function is total and  hence definable. Consequently, stability and
totality seems to be unrelated in this setting.\\
Finally, the partial completeness  result is proved using an algebraic
method.   This gives  a new  insight into  the analogy  between linear
algebra and linear  logic. We hope to get  completeness at other types
or fragments of linear logic. But the result can already be compared
with hypercoherences in which  completeness holds at first order thanks
to sequentiality~\cite{colson:hyper,longley:seq}.

\subsubsection{Acknowledgements}
I want to thank Pierre Hyvernat who was interested in the completeness
part of  this work. He  gave another proof  of the result  stated here
with an elegant combinatorial approach and found the total description
of the Por function.

\bibliographystyle{plain}
\bibliography{bib}
\vfill

\newpage

\appendix
\section{Finiteness spaces}
\label{ann:finit}
\setcounter{theorem}{\value{lc,lb,fin}}
\begin{proposition}\label{ann:lc,lb,fin}
Let $K$ be  a subspace of
  $\EvF A$.  There is an equivalence between
  \vspace{-0.5em}
  \begin{shortenumerate}
  \item\label{prop:lb} $K$ is linearly bounded,
  \item\label{prop:fin} $\supp  K=\cup\{\supp x\st x\in K\}$ is
    finitary,
  \item\label{prop:lc} the closure of $K$ is linearly compact.
  \end{shortenumerate}
\end{proposition}
\begin{proof}
  Let $K$ be linearly  bounded.  Let ${J'}\in\ort{\F(A)}$.  There is a
  finite  dimensional  subspace  $K_0$  of  $K$  such  that  $K=(K\cap
  V_{J'})\oplus  K_0$.   Since  the  dimension  of  $K_0$  is  finite,
  $\supp{K_{0}}$    is   finitary.     Besides,    $\supp   K    \cap
  {J'}=\supp{ K_0}\cap {J'}$ which is finite.  We proved that $\supp
  K\in\F(K)$.

  Let $K$ be  a subspace of $\EvF A$ such that  $\supp K$ is finitary.
  The topology induced by $\EvF  A$ on its subspace $\Field^{\supp K}$
  is the  product topology. Indeed, it is  generated by $\Field^{\supp
    K}\cap  V_{J'}=\left\{x\in\Field^{\supp   K  }  \st   \supp  x\cap
    J'=\emptyset\right\}$ with $J'\in\ort{\F(A)}$  and so $\supp K\cap
  J'$    is    finite.     By   Tychonov    Th.~\ref{th:tychonov},
  $\Field^{\supp K}$ is linearly compact.  Since the closure of $K$ is
  a closed subspace of $\Field^{\supp K}$, it is also linearly compact
  (cf.~\cite[\S 10.9(1)]{kothe:tvs1}).

  Let  $K$  be   a  linearly  compact  subspace  of   $\EvF  A$.   Let
  ${J'}\in\ort{\F(A)}$. Thanks to  the incomplete basis theorem, there
  is $D$ such  that $K= (K\cap V_{J'})\oplus D$. If  we endow $D$ with
  the discrete  topology and $K\cap V_{J'}$ with  the topology induced
  by  $K$,  then the  projections  on  each  subspace are  continuous.
  Therefore,  $D$ is  linearly  compact  as the  image  of a  linearly
  compact     space      by     a     continuous     function~\cite[\S
  10.9(2)]{kothe:tvs1}.  Besides, $D$ has  a finite dimension as every
  discrete linearly compact space.
  Thus $K$ is linearly bounded.
\end{proof}

\setcounter{theorem}{\value{separation}}
\begin{proposition}[Separation]
  \label{ann:separation}
  For every closed subspace $D$ of  $\EvF A$ and $x\notin D$, there is
  a continuous  linear form $x'\in\EvF A'$ such  that $\<x,x'\>=1$ and
  $\forall d\in D,\, \<d,x'\>=0$.
\end{proposition}
\begin{proof}
  Since  $\cap  V_{J'}=\{0\}$,  there  is $J'\in\ort{\F(A)}$  such  that
  $x\notin  V_{J'}$. Thanks  to the  incomplete basis  theorem,  we can
  define a linear form $x'$ not necessarily continuous such that
  $\<x',x\>=1$   and   $\forall   y\in  D+V_{J'},\,\<x',y\>=0$.   Since
  $V_{J'}\subset \ker{x'}$, $x'$ is continuous.
\end{proof}

\setcounter{theorem}{\value{dualseparation}}
\begin{lemma}\label{lem:finitecollection}
  Let  $F'\subset \EvF  A'$ finite.   If $0\notin\aff(F')$  then there
  exists $x\in E$ such that $\forall x'\in F',\,\<x',x\>=1$.
\end{lemma}
\begin{proof}
  Let  $x'_1,\dots,x'_n\in  F'$  be  a  maximal  linearly  independent
  collection.      We    first     prove     that    $F'     \subseteq
  \aff(x'_1,\dots,x'_n)$: Let $x'\in  F'$.  Since $(x'_i)$ is maximal,
  there  exist $\lambda_1,\dots,\lambda_n$ such  that $x'=\sum_{i=1}^n
  \lambda_i x'_i$.   Assume $0\notin  \aff(F')$, then the  equation in
  $\mu_i$:  $0=(1-\sum_{i=1}^n  \mu_{i})x'+\sum_{i=1}^n \mu_{i}x'_{i}$
  cannot  have any  solution.  By  replacing $x'$  by $\sum_i\lambda_i
  x'_i$  and   since  $(x_i)$  is  independent,  we   get  the  system
  $\left\{\left(1-\sum_{i=1}^n  \mu_i\right)\lambda_j  +\mu_j\st 1\leq
    j\leq n\right\}$  which has  no solution.  Hence,  its determinant
  $(-1)^{n-1}     (1-\sum_{i=1}^n\lambda_{i})$     is     null     and
  $\sum_{i=1}^n\lambda_{i}=1$.    Hence  $x'\in\aff(x'_1,\dots,x'_n)$.
  Since $x'_{1},\dots,x'_{n}$  are linearly independent,  there exists
  $x\in E$ such that for  any $1\leq i\leq n$, $\<x'_i,x\>=1$, hence $
  \forall x'\in F'\subseteq \aff(x'_{1},\dots,x'_{n}),\ \< x',x \>=1.$
\end{proof}

\setcounter{theorem}{\value{dualseparation}}
\begin{proposition}[Separation in the dual]
  \label{ann:dualseparation} 
  Let
  $T'$ be  a closed  affine subspace of  $\EvF A'$ such  that $0\notin
  T'$.   There   exists  $x\in  \EvF  A$  such   that  $\forall  x'\in
  T',\,\<x',x\>=1$.
\end{proposition}
\begin{proof}
  The closed subspace $T'$ does  not contain $0$, hence there exists a
  fundamental linear neighbourhood of $0$, that is $\ann(K)$ where $K$
  linearly compact in $\EvF  A$, such that $\ann(K)\cap T'=\emptyset$.
  For any  finite subspace  $F'\subseteq T'$, let  $T_{F'}=\{x\in E\st
  \forall  x'\in F',\,\<x',x\>=1\}$.   One has  $\aff(F')\subseteq T'$
  and         hence        $\aff(F')\cap\ann(K)=\emptyset$.         So
  $0\notin\aff\{x'_{|K}\st        x'\in        F'\}$.         Applying
  Lem.~\ref{lem:finitecollection}   in    $K'$,   for   every   finite
  $F'\subseteq T'$,  we get  $x\in K$ such  that $\forall  x'\in F',\,
  \<x',x\>=1$, so $x\in T_{F'}\cap K$.\\
  The collection  $(T_{F'})$ where $F'$ ranges over  finite subsets of
  $T'$  is a  filter  of closed  affine  subspaces of  $\EvF A$.   All
  elements  of   this  filter  meet  the   linearly  compact  subspace
  $K\subseteq \EvF A$. Thus
  $\textstyle K\cap \bigcap_{F'\subseteq_{\mathrm{fin}} T'} T_{F'}\neq\emptyset\,,$
  so there is $x\in K$ such that $\forall x'\in T',\ \< x',x\>=1$.
\end{proof}

\setcounter{theorem}{\value{equic}}
\begin{proposition}(Equicontinuous  spaces)\label{ann:equic}  Let $A$
  be a  relational finiteness space. A  subspace $B'$ of  $\EvF A'$ is
  linearly  bounded if and  only if  there is  $J'\in \ort{\F(A)}$  such that
  $B'\subseteq \ann_{\EvF A'}(V_{J'})$.
\end{proposition}
\begin{proof}
  First, $B'$  is linearly bounded if and  only if $\supp{B'}\in\F(A)$
  (see  Prop.~\ref{prop:lc,lb,fin}).   Second, $\exists  J'\in\F(A),\,
  B'\subseteq\ann(V_{J'})\iff\exists    J'\in\F(A),\,\supp{B'}\subseteq
  J' \iff \supp{B'}\in\F(A)$.
\end{proof}

\begin{proposition}[Reflexivity]
  \label{ann:refl}
 The map  $\iota:\EvF
  A\to\EvF A''$ defined below is a topological isomorphism.
  \begin{equation*}
    \forall x\in\EvF A,\, \iota(x):x'\in\EvF A'\mapsto \<x',x\>=x'(x).
  \end{equation*}
\end{proposition}
\begin{proof}
  If $x\neq  0$, then there  is $J\in \ort{\F(A)}$ such  that $x\notin
  V_J$.  By separation Prop.~\ref{lem:separation}, there is
  $x'\in\EvF A'$  such that $\<x',x\>=1$, hence  $\iota$ is injective.
  Let $x''\in \EvF A''$ such that $x''\neq 0$.  Let $x'\in\EvF A$ such
  that  $\<x'',x'\>=1$.    Thanks  to   separation   in  the  dual
  Prop.~\ref{dualseparation},  there is  $x\in\EvF A$  such that
  $\<x',x\>=1$    and     for    all    $y'\in\ker_{\EvF    A'}(x'')$,
  $\<y',x\>=0$. Hence $\iota(x)$ and $x''$ coincide on both $x'$ and
  $\ker(x'')$.

  If   $J'\in\ort{\F(A)}$,    then   the   support   $\supp{\ann_{\EvF
      A'}(V_{J'})}=J'$       is       in       $\ort{\F(A)}$.       By
  Prop.~\ref{prop:lc,lb,fin},   $K_{J'}=\ann_{\EvF   A'}(V_{J'})$   is
  linearly  compact in $\EvF  A'$ and  $\ann_{\EvF A''}(K_{J'})$  is a
  linear neighbourhood of $0$ in $\EvF A''$. Conversely, let $K'$ be a
  linearly     compact     subspace      of     $\EvF     A'$.      By
  Propositions~\ref{prop:equic}  and~\ref{prop:lc,lb,fin}, $\ker_{\EvF
    A}(K')$ is open in $\EvF A$.
\end{proof}

We next show that the Topological dual construction
of linear finiteness spaces coincides with the orthogonal construction
of relational finiteness spaces.
\begin{proposition}\label{ann:dual}
  The linear  finiteness space  $\EvF A '$  endowed with  the linearly
  compact open topology is  isomorphic to the linear finiteness spaces
  $\EvF{\ort A}$.
\end{proposition}
\begin{proof}
  Let  $x'\in   \Field^{\supp  A}$,  then   $x'\in\EvF{A}^\ast$  where
  $\<x',x\>=\sum_{a\in\supp  A}x'_a\,x_a$.    We  have  the  following
  equivalence:   $x'$  is  continuous   if  and   only  if   there  is
  $J'\in\ort{\F(A)}$     such    that     $\forall     x\in    V_{J'},
  \<x',x\>=\sum_{a\in\supp A}x'_ax_a=0$, that is $\supp x'\subset J'$.
  Consequently,  we   have  that   $x'\in\EvF{A}'$  if  and   only  if
  $x'\in\EvF{\ort A}$.

  Let $J\in\F(A)$.   The subspace $\Field^J$ of  $\EvF{A}$ is linearly
  compact    thanks   to    Prop.~\ref{prop:lc,lb,fin}.    Hence
  $\ann(\Field^J)=\{x'\st  \supp{ x'}  \in  \ort{\F(A)},\,\supp x'\cap
  J=\emptyset\}$  is open in  $\EvF{A}'$.  Conversely,  if $K\subseteq
  \EvF   A$  is  linearly   compact,  then   $\supp  K$   is  finitary
  (Prop.~\ref{prop:lc,lb,fin})  and  $V_{\supp  K}$ is  a  fundamental
  linear neighbourhood of $0$. Moreover, $V_{\supp K}\subseteq\ann(K)$
  hence  $\ann(K)$ is  open  in  $\EvF{A}'$. We  proved  that the  two
  topologies coincide.
\end{proof}

\section{Interpretation of proofs in $\LF$.}
\label{annex:interp}

Hypocontinuity can be generalised to multilinear functions:
\begin{definition}\label{def:hypo}
  Let  $(A_i)_{i\leq   n}$  be  a  finite   collection  of  relational
  finiteness spaces.   An $n$-linear form  $\phi:\times_i \EvF{A_i}\to
  \Field$ is  \vip{hypocontinuous} if  for any $(K_i)$  collection of
  linearly compact  subspaces of $\EvF{A_i}$s  (respectively), for any
  $i_0$  there exists  a fundamental  linear neighbourhood  $ U_{i_0}$
  such that  $\phi(\times A_i)=0$ where  $A_i=K_i$ if $i\neq  i_0$ and
  $A_{i_0}=U_{i_0}$.
\end{definition}

Any proof of the sequent  $\vdash\Gamma$ of formula of MELL+MIX+SUM is
interpreted  as  a  continuous   linear  form  on  $\sem{\Gamma}$.  If
$\Gamma=A_1,\dots,A_n$ then a proof of the sequent can be expressed as
a  hypocontinuous  $n$-linear  on  $\times_i\sem{A_i}'$.  Finally,  if
$\Gamma=\Gamma_1,A$, then a proof of this sequent can be equivalently
 described  as a continuous linear  function from $\sem{\Gamma_1}$
to  $\sem{A}$.  We freely  use  these  different  presentations in  to
describe proofs.

The  interpretation  of  proofs   of  MELL+MIX+SUM  are  described  on
Fig.~\ref{fig:proofs1}-~\ref{fig:proofs}  and   is  similar  with  the
presentation of~\cite{lafont:gsl} except for exponentials.


\begin{figure}[ht!]
  \caption{Interpretation of proofs of MELL+MIX+SUM in $\LF$}\label{fig:proofs1}
  \centering 
  \subfigure{
    \begin{minipage}[c]{.4\linewidth}
      \centering \vip{A proof $\pi$}
    \end{minipage}\vline\hskip.5em
    \begin{minipage}[l]{.59\linewidth}
      \centering\vip{ Its interpretation}
    \end{minipage}
  }   \\\hrule
  \subfigure{
    \begin{minipage}[c]{.4\linewidth}
      \begin{prooftree}
        \AxiomC{}  \RightLabel{\small AX}  \UnaryInfC{$\jug  A,\, \ort
          A$}
      \end{prooftree}
    \end{minipage}
    \vline\hskip.5em
    \begin{minipage}[l]{.59\linewidth}
      $\begin{array}{rccccl}
        \sem\pi: &\ \sem A'&\times& \sem A&\to& \Field\\
        &x'&,&x&\mapsto& \<x',x\>
      \end{array}$
    \end{minipage}
  }
  \hrule \subfigure{
    \begin{minipage}[c]{.4\linewidth}
      \begin{prooftree}
        \AxiomC{$\rho_1$}  \dottedLine \UnaryInfC{$\jug  \Gamma_1, A$}
        \dottedLine \AxiomC{$\rho_2$}  \UnaryInfC{$\jug \Gamma_2, \ort
          A$}      \RightLabel{\small      CUT}      \BinaryInfC{$\jug
          \Gamma_1,\Gamma_2$}
      \end{prooftree}
    \end{minipage}\vline\hskip.5em
    \begin{minipage}[c]{.59\linewidth}
      $\begin{array}{ccccccl}
        \multicolumn{7}{l}{\sem{\rho_1}\in\sem{\ort\Gamma_1\multimap
            A}\qquad\sem{\rho_2}\in \sem{\ort \Gamma_2\multimap \ort
            A}}\\\\
        \sem{\pi}&:&\sem{\Gamma_1}'&\times &\sem{\Gamma_2}'&\to &\Field\\
        &&\gamma_1'&,&\gamma_2'&\mapsto&
        \<\rho_2(\gamma_2'),\rho_1(\gamma_1')\>
      \end{array}
      $
    \end{minipage}}
  \\\hrule \subfigure{
    \begin{minipage}[c]{.4\linewidth}
      \begin{prooftree}
        \AxiomC{$\rho_i$}   \dottedLine
        \UnaryInfC{$\jug   \Gamma$}   \RightLabel{$\sum\limits_{i=1}^n
          a_i=1$} 
        \UnaryInfC{$\jug \Gamma$}
      \end{prooftree}
    \end{minipage}\vline\hskip.5em
     \begin{minipage}[l]{.59\linewidth}
       $\begin{array}{rccl}
         \multicolumn{4}{l}{\sem{\rho_i}\in\sem{\Gamma}'\qquad
         \sum\limits_{i=1}^na_i=1}\\\\
         \sem{\pi}:&\sem{\Gamma}&\to&\Field\\
         &\gamma&\mapsto& \sum\limits_{i=1}^n a_i\,\<\sem{\rho_i},\gamma\>
       \end{array}$
     \end{minipage}
  }
  \\\hrule \subfigure{
    \begin{minipage}[c]{.4\linewidth}
      \begin{prooftree}
        \AxiomC{}         \RightLabel{\small         $\text{MIX}_{0}$}
        \UnaryInfC{$\jug \bot$}
      \end{prooftree}
    \end{minipage}\vline\hskip.5em
    \begin{minipage}[l]{.59\linewidth}
      $\begin{array}{rccl}
        \sem\pi: &\ \Field&\to& \Field\\
        &a&\mapsto& a
      \end{array}$
    \end{minipage}} \\\hrule \subfigure{
    \begin{minipage}[c]{.4\linewidth}
      \begin{prooftree}
        \AxiomC{$\rho_1$}   \dottedLine   \UnaryInfC{$\jug  \Gamma_1$}
        \dottedLine    \AxiomC{$\rho_2$}    \UnaryInfC{$\jug\Gamma_2$}
        \RightLabel{\small MIX} \BinaryInfC{$\jug \Gamma_1, \Gamma_2$}
      \end{prooftree}
    \end{minipage}\vline\hskip.5em
    \begin{minipage}[l]{.59\linewidth}
      $\begin{array}{rccccl}
        \multicolumn{6}{c}{\sem{\rho_1}\in\sem{\Gamma_1}'
          \qquad\sem{\rho_2}\in \sem{\Gamma_2}'}\\\\
        \sem\pi: &\ \sem{\Gamma_1}'&\times&\sem{\Gamma_2}'&\to& \Field\\
        &\gamma_1'&,&\gamma_2'&\mapsto&            \sem{\rho_1}(\gamma_1')\,
        \sem{\rho_2}(\gamma_2')
      \end{array}$
    \end{minipage}}
\\\hrule
\subfigure{
    \begin{minipage}[c]{.4\linewidth}
      \begin{prooftree}
        \AxiomC{$\rho$}        \dottedLine        \UnaryInfC{$\Gamma$}
        \RightLabel{$\bot$} \UnaryInfC{$\jug\Gamma,\bot $}
      \end{prooftree}
    \end{minipage}\vline\hskip.5em
    \begin{minipage}[l]{.59\linewidth}
      $\begin{array}{rccccl}
        \multicolumn{6}{l}{\sem{\rho}\in\sem{\Gamma}'}\\\\
        \sem\pi: &\ \sem{\Gamma}'&\times&\Field&\to& \Field\\
        &\gamma&,&a&\mapsto& \sem\rho(\gamma')
      \end{array}$
    \end{minipage}}
\end{figure} 
\begin{figure}[ht!]
 \caption{Interpretation of proofs of MELL+MIX+SUM in $\LF$}\label{fig:proofs}
  \centering
 \subfigure{
    \begin{minipage}[c]{.4\linewidth}
      \begin{prooftree}
        \AxiomC{$\rho$}   \dottedLine   \UnaryInfC{$\jug  \Gamma,A,B$}
        \RightLabel{$\parr$} \UnaryInfC{$\jug\Gamma,A\parr B$}
      \end{prooftree}
    \end{minipage}\vline\hskip.5em
    \begin{minipage}[l]{.59\linewidth}
      $\begin{array}{rccccl}
        \multicolumn{6}{l}{\sem{\rho}\in\sem{\ort\Gamma\multimap
            \left({A}\parr{B}\right)}} \\\\
        \sem\pi:&\sem{\Gamma}'&\times& (\sem{A}\lparr\sem{B})'&\to&\Field\\
        &\ \gamma'&,&\phi& \mapsto& \<\phi,\sem\rho(\gamma')\>
      \end{array}$
    \end{minipage}}
  \\\hrule \subfigure{
    \begin{minipage}[c]{.4\linewidth}
      \begin{prooftree}
        \AxiomC{} \RightLabel{$1$} \UnaryInfC{$\jug 1$}
      \end{prooftree}
    \end{minipage}\vline\hskip.5em
    \begin{minipage}[l]{.59\linewidth}
      $\begin{array}{rccl}
        \sem\pi: &\ \Field&\to& \Field\\
        &a&\mapsto& a
      \end{array}$
    \end{minipage}} \\\hrule \subfigure{
    \begin{minipage}[c]{.4\linewidth}
      \begin{prooftree}
        \AxiomC{$\rho_1$}  \dottedLine \UnaryInfC{$\jug \Gamma_1,A_1$}
        \dottedLine  \AxiomC{$\rho_2$}  \UnaryInfC{$\jug\Gamma_2,A_2$}
        \RightLabel{$\otimes$}  \BinaryInfC{$\jug  \Gamma_1, \Gamma_2,
          A\otimes B$}
      \end{prooftree}
    \end{minipage}\vline\hskip.5em
    \begin{minipage}[l]{.59\linewidth}
      $\begin{array}{rccccccl}
        \multicolumn{8}{l}{\sem{\rho_1}\in\sem{\ort\Gamma_1\multimap
            A_1}
          \qquad\sem{\rho_2}\in \sem{\ort\Gamma_2\multimap A_2}}\\\\
        \sem\pi:&\ \sem{\Gamma_1}'&\times&
        \sem{\Gamma_2}'&\times &(\sem A'\lparr\sem B)&\to&\Field\\
        &\   \gamma_1'&,&\gamma'_2&,&\multicolumn{3}{c}{\phi   \mapsto
          \phi(\sem{\rho_1}(\gamma'_1),\sem{\rho_2}(\gamma'_2)) }
      \end{array}$
    \end{minipage}}
  \\\hrule \subfigure{
    \begin{minipage}[c]{.4\linewidth}
      \begin{prooftree}
        \AxiomC{} \RightLabel{$\top$} \UnaryInfC{$\jug \Gamma, \top$}
      \end{prooftree}
    \end{minipage}\vline\hskip.5em
    \begin{minipage}[l]{.59\linewidth}
      $\begin{array}{rccccl}
        \sem\pi: &\ \sem{\Gamma}'&\times& \{0\}&\to& \Field\\
        &\gamma'&,&0&\mapsto& 0
      \end{array}$
    \end{minipage}} \\\hrule \subfigure{
    \begin{minipage}[c]{.4\linewidth}
      \begin{prooftree}
        \AxiomC{$\rho_1$}  \dottedLine \UnaryInfC{$\jug  \Gamma, A_1$}
        \AxiomC{$\rho_2$}  \dottedLine \UnaryInfC{$\jug  \Gamma, A_2$}
        \RightLabel{$\with$} \BinaryInfC{$\jug \Gamma, A_1\with A_2$}
      \end{prooftree}
    \end{minipage}\vline\hskip.5em
    \begin{minipage}[l]{.59\linewidth}
      $\begin{array}{rccccl}
        \multicolumn{6}{l}{\sem{\rho_1}\in\sem{\ort\Gamma\multimap
            A_1}
          \qquad\sem{\rho_2}\in \sem{\ort\Gamma\multimap A_2}}\\\\
        \sem\pi:&\ \sem{\Gamma}'&\times& (\sem{A_1}'\oplus \sem{A_2}')&\to&\Field\\
        &     \multicolumn{5}{l}{   \gamma',x_1'+x_2'        \mapsto
          \<x'_1,\sem{\rho_1}(\gamma')\>+\<x'_2,\sem{\rho_2}(\gamma')\>}
      \end{array}$
    \end{minipage}} \\\hrule \subfigure{
    \begin{minipage}[c]{.4\linewidth}
      \begin{prooftree}
        \AxiomC{$\rho$}  \dottedLine  \UnaryInfC{$\jug  \Gamma,  A_1$}
        \RightLabel{$\oplus_g$} \UnaryInfC{$\jug \Gamma, A_1\oplus A_2$}
      \end{prooftree}
    \end{minipage}\vline\hskip.5em
     \begin{minipage}[l]{.59\linewidth}
      $\begin{array}{rccccl}
        \multicolumn{6}{l}{\sem{\rho}\in\sem{\ort\Gamma\multimap
            A_1}}\\\\
        \sem{\pi}:&\sem{\Gamma}'&\times& (\sem{A_1}'\times\sem{A_2}')\\
        &\gamma&,&(x'_1,x'_2)&\mapsto& \<x'_1,\sem\rho(\gamma)\>
       \end{array}$
     \end{minipage}
  }
  \\\hrule
  \subfigure{
    \begin{minipage}[c]{.4\linewidth}
      \begin{prooftree}
        \AxiomC{$\rho$} \dottedLine \UnaryInfC{$\jug \Gamma$}
        \RightLabel{Weak} \UnaryInfC{$\jug \Gamma,\wn\ort A$}
      \end{prooftree}
    \end{minipage}\vline\hskip.5em
     \begin{minipage}[l]{.59\linewidth}
       $\begin{array}{rccccl}
         \multicolumn{6}{l}{\sem{\rho}\in\sem{\Gamma}'}\\\\
         \sem{\pi}:&\sem{\Gamma}&\times&\sem{\oc A}&\to&\Field\\
         &\gamma&,&d&\mapsto& \sem\rho(\gamma)\ \mathbbm 1_{d=e_{[]}}
       \end{array}$
     \end{minipage}
  }
  \\\hrule \subfigure{
    \begin{minipage}[c]{.4\linewidth}
      \begin{prooftree}
        \AxiomC{$\rho$}  \dottedLine \UnaryInfC{$\jug  \Gamma, \wn\ort
          A,\wn\ort A$}
        \RightLabel{contr}   \UnaryInfC{$\jug
          \Gamma,\wn \ort A$}
      \end{prooftree}
    \end{minipage}\vline\hskip.5em
     \begin{minipage}[l]{.59\linewidth}
       $\begin{array}{rccccl}
         \multicolumn{6}{l}{\sem{\rho}\in\sem{\Gamma}\lparr\sem{\oc A}\lparr\sem{\oc
             A}}\\\\
         \sem{\pi}:&\sem{\Gamma}'&\times&\sem{\oc A}&\to&\Field\\
         &\gamma'&,&\multicolumn{3}{l}{d\mapsto\sum\limits_{d_1,d_2}
         \sem\rho(\gamma',d_1,d_2)\ \mathbbm 1_{d=d_1\ltens d_2}}
       \end{array}$
     \end{minipage}
  }
   \\\hrule \subfigure{
    \begin{minipage}[c]{.4\linewidth}
      \begin{prooftree}
        \AxiomC{$\rho$} \dottedLine  \UnaryInfC{$\jug \Gamma,\ort  A$}
        \RightLabel{Der} \UnaryInfC{$\jug \Gamma,\wn \ort A$}
      \end{prooftree}
    \end{minipage}\vline\hskip.5em
     \begin{minipage}[l]{.59\linewidth}
       $\begin{array}{rccccl}
         \multicolumn{6}{l}{\sem{\rho}\in\sem{\Gamma}\lparr\sem{A}}\\\\
         \sem{\pi}:&\sem{\Gamma}'&\times&\sem{\oc A}&\to&\Field\\
         &\gamma'&,&X&\mapsto& \sum\limits_{\#\mu=1}\sem\rho(\gamma',X_\mu)
       \end{array}$
     \end{minipage}
  }
   \\\hrule \subfigure{
    \begin{minipage}[c]{.4\linewidth}
      \begin{prooftree}
        \AxiomC{$\rho$}     \dottedLine    \UnaryInfC{$\jug    \wn\ort
          A_1,\dots,\wn\ort A_n,\, B$}
        \RightLabel{Prom}   \UnaryInfC{$\jug
         \wn\ort A_1,\dots,\wn\ort A_n,\oc B$}
      \end{prooftree}
    \end{minipage}\vline\hskip.5em
     \begin{minipage}[l]{.59\linewidth}
       $\begin{array}{rccl}   \multicolumn{4}{l}{\sem{\rho}\in\Poln  n
           {\sem{A_1} \times \dots
             \times \sem{A_n};\sem{A}}}\\\\
         \sem{\pi}:&\Poln n {\sem{A_1} \times \dots
           \times \sem{A_n};\sem{\oc A}}\\
         &\multicolumn{3}{l}{x_1,\dots,x_n\mapsto
           \exp{\left(\sem{\rho}(x_1,\dots,x_n)\right)}}
       \end{array}$
     \end{minipage}
  }
\end{figure}

\clearpage

\section{Totality candidate characterisation.}\label{annex:totalite}

\setcounter{theorem}{\value{aff}}
\begin{corollary}\label{ann:aff}
  Let $\T$  be a subset  of $\EvF A$.  If  $\pol\T\neq\emptyset$, then
  $\dpol\T=\caff(\T)$.
\end{corollary}
\begin{proof}
  Since $\T\subseteq\dpol\T$  and $\dpol\T$  is affine close,  we have
  $\T\subseteq\caff(\T)\subseteq\dpol\T$.                 Consequently,
  $\pol\T\subseteq\pol{[\,\caff(\T)\,]}      \subseteq\pol\T$      and
  $\dpol{[\caff(\T)]}= \dpol{\T}$.  Moreover, if $\pol\T\neq\emptyset$
  then  there   is  $x'\in\EvF  A'$   such  that  for   any  $x\in\T$,
  $\<x',x\>=1$,  and  so for  any  $x\in\caff(\T)$, $\<x',x\>=1$.   We
  infer  that $0\notin  \caff(\T)$.  Thanks  to  the characterisation
  Lem.~\ref{prop:charac},          $\dpol{[\,\caff(\T)\,]}=\caff(\T)$.
  Finally, $\caff(\T)=\dpol\T$.
\end{proof}

\setcounter{theorem}{\value{mult}}
\begin{proposition}\label{ann:mult}
  \begin{eqnarray}
    \T(A\otimes       B)&=&\caff(\T(A)\otimes\T(B)),\nonumber\\
    \T(A\multimap B) &=& \left\{f\in\EvF A\st f(\T(A))\in\T(B)\right\}.\label{eq:fun}
  \end{eqnarray}
\end{proposition}
\begin{proof}
  If   $x'\in\pol{\T(A)}$  and  $y'\in\pol{\T(B)}$,   then  $x'\otimes
  y':(x,y)\to  \<x',x\>\<y',y\>$  is  in  $\pol{[\T(A)\otimes\T(B)]}$.
  Thanks to Cor.~\ref{cor:aff},
  $\caff(\T(A)\otimes\T(B))=\dpol{[\T(A)\otimes\T(B)]}$.

  The     second    equation     comes    from     the    equivalences:\\
  $\begin{array}{rcl}
    f\in\pol{[\T(A)\otimes\pol{\T(B)}]}&\iff&                     \forall
    x\in\T(A),\,\forall y'\in\pol{\T(B)},\,\<f(x),y'\>=1,\\
    &\iff&   \forall
    x\in\T(B),\,f(x)\in\T(B).
  \end{array}$
\end{proof}

\begin{proposition}
  \begin{equation*}
    \T(A\oplus B)=\caff{(\T(A)\times\ker(\pol{\T(B)})\cup\ker(\pol{\T(A)})\times \T(B))}
  \end{equation*}
\end{proposition}
\begin{proof}
  By  construction,\\$\T(A  \oplus  B)=\{(x,  y) \mid\  \forall  u'\in
  \pol{\T(A)},\    v'\in\pol{\T(B)},\   \<x,u'\>+\<y,v'\>=1\}$.\\   So
  $\caff(\T(A)\times
  \ker(\pol{\T(B)})\cup\ker(\pol{\T(A)})\times\T(B))\subseteq
  \T(A\oplus  B)$.   Reciprocally, let  $z=(x,  y)\in \T(A\oplus  B)$,
  $u'_0\in      \pol{\T(A)}$      and      $v'_0\in      \pol{\T(B)}$,
  $D(A)=\ldir(\pol{\T(A)})$   and   $D(B)=\ldir(\pol{\T(B)})$.    Then
  $\pol{\T(A)}=u'_0+D(A)$ and  $\pol{\T(B)}=v'_0+D(B)$.  Therefore for
  all  $d_x'\in  D(A),\,\<d_x',x\>  =  0$;  for  all  $d_y'\in  D(B)$,
  $\<d_y',y\>   =  0  $   and  $\<u'_0,x\>   +\<v'_0,y\>  =   1$.   If
  $\<v'_0,y\>=0$    then   $y\in\ker(D(B))$   and    $x\in\T(A)$,   so
  $z=(x,y)\in\T(A)\times    \ker(\pol{\T(A)})$    (respectively,    if
  $\<u'_0,x\>=0$,    then   $z\in\ker(\pol{\T(B)})\times\T(B)$).    If
  $\<u'_0,x\>\neq     0$      and     $\<v'_0,y\>\neq     0$,     then
  $z=\<u'_0,x\>(\frac{x}{\<u'_0,x\>},0)+\<v'_0,y\>(0,
  \frac{y}{\<v'_0,y\>})$.                                            So
  $z=(x,y)\in\caff(\T(A)\times\ker(\pol{\T(B)})\cup\ker(\pol{\T(A)})\times
  \T(B))$.
\end{proof}

\begin{proposition}\label{ann:exp}
  \begin{eqnarray}
    \T(\oc A)&=&\caff(\exp x\st x\in\T(A))\label{eq:exp},\\
    \T(\wn    A)&=&\left\{F\in\Polc{\EvF    A}\st   \forall
      x\in\T(A),\,F(x)=1\right\}\label{eq:wn},\\
    \T(\oc A\multimap B)&=&\left\{F\in\Polc{\EvF A,B}\st \forall
      x\in\T(A),\,F(x)\in\T(B)\right\}\label{eq:kleisli}.
  \end{eqnarray}
\end{proposition}
\begin{proof}
  Let $1\in\Polk{\EvF A}$ be  the constant function: $\forall x\in\EvF
  A,\,\<\exp    x,1\>=1$.    We   have    $1\in\pol{\left\{\exp   x\st
      x\in\T(A)\right\}}$. Therefore,  the first equality
  comes from Cor.~\ref{cor:aff}.
  Using~\eqref{eq:fun}  and~\eqref{eq:exp}, we get~\eqref{eq:kleisli}.
  The   equality~\eqref{eq:wn}    comes   from   the    linear   logic
  equivalence:~$\wn A\simeq \oc(\ort A)\multimap 1$.
\end{proof}

\end{document}